\documentclass[lettersize,journal]{IEEEtran}
\usepackage{amsmath,amsfonts}
\usepackage{array}
\usepackage[caption=false,font=normalsize,labelfont=sf,textfont=sf]{subfig}
\usepackage{cite}
\usepackage{amsmath,amssymb,amsfonts}
\usepackage{algorithm}
\usepackage{algorithmicx}
\usepackage{algpseudocode}
\usepackage{graphicx}
\usepackage{bm}
\usepackage{cases}
\usepackage{textcomp}
\usepackage{color}
\newtheorem{remark}{Remark}
\newtheorem{assumption}{Assumption}
\newtheorem{definition}{Definition}
\newtheorem{proposition}{Proposition}
\newtheorem{theorem}{Theorem}
\newtheorem{lemma}{Lemma}
\newtheorem{problem}{Problem}
\newenvironment{proof}{{\indent \indent \it Proof:}}{\hfill $\blacksquare$\par}
\def\BibTeX{{\rm B\kern-.05em{\sc i\kern-.025em b}\kern-.08em
    T\kern-.1667em\lower.7ex\hbox{E}\kern-.125emX}}
\begin{document}
\title{Propensity Formation-Containment Control of Fully Heterogeneous Multi-Agent Systems via Online Data-Driven Learning}

\author{Ao Cao, Fuyong Wang, and Zhongxin Liu, \IEEEmembership{Member, IEEE}
	\thanks{This work was supported by the National Natural Science Foundation of China under Grant 62103203. (Corresponding author: Fuyong Wang.)}
	\thanks{Ao Cao, Fuyong Wang and Zhongxin Liu  are with the College of Artificial Intelligence, Nankai	University, Tianjin 300350, China (e-mail: caoao@mail.nankai.edu.cn;  wangfy@nankai.edu.cn; lzhx@nankai.edu.cn).}}
\maketitle

\begin{abstract}
This paper introduces an online data-driven learning scheme designed to address a novel problem in propensity formation and containment control for fully heterogeneous multi-agent systems. Unlike traditional approaches that rely on the eigenvalues of the Laplacian matrix, this problem considers the determination of follower positions based on propensity factors released by leaders. To address the challenge of incomplete utilization of leader information in existing multi-leader control methods, the concept of an influential transit formation leader (ITFL) is introduced. An adaptive observer is developed for the agents, including the ITFL, to estimate the state of the tracking leader or the leader's formation. Building on these observations, a model-based control protocol is proposed, elucidating the relationship between the regulation equations and control gains, ensuring the asymptotic convergence of the agent's state. To eliminate the necessity for model information throughout the control process, a new online data-driven learning algorithm is devised for the control protocol. Finally, numerical simulation results are given to verify the effectiveness of the proposed method.
\end{abstract}
\section{Introduction}
 The rapid progress in communication, perception, and computer technologies has catalyzed the evolution of human-engineered and composite systems towards largescale and intelligence. These extensive systems typically comprise a multitude of interconnected intelligent subsystems working collaboratively to achieve predefined objectives \cite{wooldridge2009introduction,baillieul2007control,cao2012overview,9842347}. Presently, the formation-containment control (FCC) problem within multi-agent systems (MAS) has gained significant prominence in the realm of distributed control systems. FCC is concerned with assisting leaders in forming desired configurations while guiding followers to enter the convex hull delineated by the leaders. Given its specific control objectives, FCC finds a wide array of applications in practical systems, including unmanned airship systems \cite{yu2021distributed}, unmanned air-ground vehicle systems \cite{cheng2023data}, rigid spacecraft systems \cite{cui2020truly}, and unmanned aerial vehicle systems \cite{ouyang2021neural}.

Significant theoretical progress has been achieved in addressing the FCC problem for homogeneous systems, where both the agent dynamics and the formation dynamics of leaders are identical, as evidenced by several studies \cite{dong2015formation,9409728,9870036}. In recent years, attention has also turned to the semi-heterogeneous FCC problem, which deals with systems featuring heterogeneous agent dynamics and homogeneous formation dynamics. These research outcomes span various categories, ranging from undirected \cite{WANG201826} to directed topologies \cite{8384027}, from distributed \cite{WANG2017392} to fully distributed solutions \cite{8579593}, and from fixed \cite{GAO2021146} to switched topologies \cite{cheng2023data}. Many of these investigations rely on output regulation equations \cite{CAI2017299,6026912,DENG201962} and assume that leaders have identical formation dynamics. However, practical applications often exhibit diverse requirements, and certain scenarios demand heterogeneous formation dynamics to effectively tackle complex tasks. For instance, collaborative search and rescue missions involving unmanned air-ground vehicles after a disaster \cite{7330001}. Despite its significance, there remains a noticeable research gap when it comes to fully heterogeneous FCC problems that involve distinct formation dynamics and agent dynamics.

Typically, achieving FCC in heterogeneous systems necessitates the development of distributed observers for followers, enabling them to estimate convex combinations of multiple leaders \cite{HAGHSHENAS2015210,8277155,8869852,9758951}. This becomes especially challenging when the leader systems exhibit heterogeneity. Recent work in \cite{9627528} and \cite{9802518} introduced adaptive distributed observers for the containment control (CC) problem \cite{10336935} in directed and undirected graphs, respectively.  These approaches relaxed the constraint that leader dynamics must be identical. However, they come with an implicit requirement regarding the communication graph: all formation leaders must communicate directly with at least one follower. These methods become inapplicable when a formation leader needs to transmit information through an influential transit formation leader (ITFL), where the followers do not utilize the information from this leader.  Additionally, it's worth noting that in nearly all existing results, the convergence position of the followers is influenced by the eigenvalues of the Laplacian matrix. This dependency is neither reasonable nor practical in real-world operational environments.

It is important to highlight that the results discussed regarding FCC and CC in heterogeneous systems rely on precise model knowledge to solve the output regulation equations. However, when acquiring accurate models through explicit modeling becomes a challenge, the design of controllers using data-driven methods that exclusively rely on measurement data becomes necessary. Considerable advancements have been made in the development of data-driven algorithms across various domains, including optimal control \cite{8703172,10061542,9304046}, safety control \cite{10101826,8970529}, predictive control \cite{8039204,9109670,BRESCHI2023110961} and robust control \cite{KARIMI2017227,8957584}, with many of these techniques drawing inspiration from reinforcement learning (RL) and behavioral theory. Despite the rapid progress of data-driven control methods in recent years, the cooperative control problem within fully heterogeneous multi-agent systems, relying solely on measurement data, remains an active area of research that has yet to be fully addressed.

The primary objective of this article is to develop formation-containment controllers tailored for fully heterogeneous MAS by exclusively leveraging measurement data. The aim is to tackle two critical challenges: the underutilization of leader information in specific topological scenarios and the reliance of follower convergence positions on topology. To address these challenges, this paper introduces the concept of Propensity Formation-Containment Control (PFCC) for fully heterogeneous MAS. The approach assumes that formation leaders generate propensity factors to steer the motion of the followers, which encompasses the development of a data-based adaptive observer tailored for pertinent agents, notably the Influential Transit Formation Leader (ITFL), with the purpose of gathering crucial information. Building upon this observer, the paper develops a model-based PFCC controller and a corresponding online data-driven learning algorithm. The main contributions of this paper can be summarized as follows:

1) A novel Propensity Formation-Containment Control (PFCC) problem is introduced for fully heterogeneous Multi-Agent Systems (MAS) within directed graphs. In contrast to prior Formation-Containment Control (FCC) findings \cite{cheng2023data,dong2015formation,9409728,9870036,WANG201826,8384027,WANG2017392,8579593,GAO2021146}, this problem allows for nonidentical formation dynamics among leaders, and the convergence positions of followers in this context are determined by the propensity factors released by the leaders rather than relying on the eigenvalues of the Laplacian matrix.

2) We have designed a distributed adaptive observer reliant on measurement data for each agent. Notably, this observer offers distinct advantages over the findings in \cite{HAGHSHENAS2015210,8277155,8869852,9627528,9802518,WANG201826,8384027,WANG2017392,8579593,GAO2021146,cheng2023data,9758951} as it doesn't necessitate any dynamic knowledge and effectively addresses the issue of underutilizing leader information in specific communication topologies.

3) We have developed a novel online data-driven learning scheme for addressing the PFCC problem within fully heterogeneous dynamics. This  scheme guarantees both asymptotic convergence and stability. Moreover, it is versatile, applicable not only to under-actuated and fully-actuated systems but also to over-actuated systems.

The rest of the paper is organized as follows. In Section \ref{S2}, we provide an introduction to the necessary background and delve into the PFCC problem within a directed graph. Section \ref{S3} is dedicated to the design of data-based adaptive observers for agents to facilitate the acquisition of valid information. We also present the design of both model-based and data-driven learning methods for the PFCC problem. In Section \ref{S4}, we offer an illustrative example to validate the efficacy of the proposed algorithm. Finally, Section \ref{S5} offers a summary of the paper.

\section{Preliminaries}\label{S2}
\subsection{Notation}
In this paper, $\bm{1}_n$ represents the $n$-dimensional column vector with all entries of one. $\bm{0}_{n\times m}$ denotes a zero matrix with $n$ rows and $m$ columns. $\bm{I}_n$ and $\bm{0}_n$ denote the identity matrix and zero matrix of order $n$, respectively. $\lambda _{\min }(\cdot)$, $\lambda _{\max }(\cdot)$, $\sigma _{\max }(\cdot)$ and $\rho(\cdot)$ stand for the minimum eigenvalue, maximum eigenvalue, maximum singular value and spectral radius of the matrix, respectively. For matrices $P_1, P_2,..., P_n$ with appropriate dimensions, let $col(P_1, P_2,..., P_n)$ and $diag(P_1, P_2,..., P_n)$ denote the matrix $[P_1^T,P_2^T,..., P_n^T]^T$ and diagonal matrix, respectively. Moreover, for a column vector $\mathcal{D}=[\mathcal{D}_1,\mathcal{D}_2,...,\mathcal{D}_n]^T\in \mathbb{R}^n$, $g\in \mathbb{R}^{T-t+1}$
\begin{align*}
	&\text{vecv}(\mathcal{D})=[\mathcal{D}_1^2, \sqrt{2}\mathcal{D}_1\mathcal{D}_2,...,\sqrt{2}\mathcal{D}_1\mathcal{D}_n,\\ \nonumber
	&~~~~~~~~~~~~~~~~~~~~~\mathcal{D}_2^2, \sqrt{2}\mathcal{D}_2\mathcal{D}_3,..., \sqrt{2}\mathcal{D}_{n-1}\mathcal{D}_n, \mathcal{D}_n^2]^T. 
\end{align*}
For a symmetric matrix $\mathcal{S}=[\mathcal{S}_{ij}]\in \mathbb{R}^{n\times n}$,
\begin{align*}
	&\text{vecm}(\mathcal{S})=[\mathcal{S}_1, \sqrt{2}\mathcal{S}_{12},...,\sqrt{2}\mathcal{S}_{1n},\\ \nonumber
	&~~~~~~~~~~~~~~~~~~~~~~~\mathcal{S}_2, \sqrt{2}\mathcal{S}_{23},..., \sqrt{2}\mathcal{S}_{(n-1)n}, \mathcal{S}_{nn}]^T.
\end{align*}
 For an arbitrary matrix $\mathcal{M}\in \mathbb{R}^{m\times n}$, $\text{vec}(\mathcal{M})=[\mathcal{M}_1^T,\mathcal{M}_2^T,...,\mathcal{M}_n^T]^T$, where $\mathcal{M}_i$ is the $i$-th column of $\mathcal{M}$.
\subsection{Background Knowledge}
Consider a directed network topology $\mathcal{G} = (\mathcal{V},\mathcal{E})$ composed of a node set $\mathcal{V} = \{ {v_1},{v_2},...,{v_p}\}$ and an edge set ${\mathcal{E}} = \{ ({v_i},{v_j})|i,j \in \{1,2,...,p\}\}  \subseteq \mathcal{V} \times \mathcal{V}$. The adjacency matrix of $\mathcal{G}$ is defined as $\mathcal{A} = [{a_{ij}}] \in {\mathbb{R}^{p \times p}}$, where ${a_{ij}} = 1$ if $({v_j},{v_i}) \in{\mathcal{E}} $ and ${a_{ij}}=0$ otherwise. The neighbor set of node $v_i$ is denoted by ${\mathcal{N}_i} = \{ {v_j}|({v_j},{v_i}) \in {\mathcal{E}}\}$. The in-degree matrix is defined as $\mathcal{D} = diag\{ {d_1,d_2,...,d_p}\}  \in {\mathbb{R}^{p \times p}}$ with ${d_i} = \sum\nolimits_{j \in {\mathcal{N}_i}} {{a_{ij}}}$. The Laplacian matrix $\mathcal{L}$ of $\mathcal{G}$ can be expressed as $\mathcal{L} = \mathcal{D} - \mathcal{A}$.  

In this paper, we consider a group of $N + M + 1$ agents, consisting of one tracking leader, $M$ formation leaders, and $N$ followers.  Denote the formation leader set and follower set as $\mathbb{F} \buildrel \Delta \over = \{1,2,...,N\}$ and $\mathbb{L} \buildrel \Delta \over = \{N+1,N+2,...,N+M\}$, respectively. The tracking leader is designated as index $0$. To delineate the relationships between various types of agents, we use $g^q_i$ to represent the edge weight from formation leader $q$ to follower $i$, and $g^0_q$ to represent the edge weight from tracking leader to formation leader $q$. In addition, let $\bar g_i^q=1$ if formation leader $q$ to follower $i$ has a directed path and $\bar g_i^q=0$ otherwise. Likewise, $\bar{a}_{qm} = 1$ if there is a directed path from formation leader $m$ to formation leader $q$ and $\bar a_{qm}=0$ otherwise. Note that $\bar g_i^q$ and $\bar a_{qm}$ are unknown to followers and formation leaders. 

The Laplacian matrix $\mathcal{L}\in {\mathbb{R}^{{(N+M+1)} \times {(N+M+1)}}}$ associated with the network topology $\mathcal{G}$ of all agents can be divided into
\begin{align*}
	{\cal L} = \left[ {\begin{array}{*{20}{c}}
			0&{{\bm{0}_{1 \times N}}}&{{\bm{0}_{1 \times M}}}\\
			{{\bm{0}_{N \times 1}}}&{{{\cal L}_1}}&{{{\cal L}_2}}\\
			{{{\cal L}_0}}&{{\bm{0}_{M \times N}}}&{{{\cal L}_3}}
	\end{array}} \right],
\end{align*}
where the dimensions of the sub-blocks are ${\mathcal{L}_0} \in {\mathbb{R}^{M \times 1}}$, ${\mathcal{L}_1} \in {\mathbb{R}^{N \times N}}$, ${\mathcal{L}_2} \in {\mathbb{R}^{N \times M}}$ and ${\mathcal{L}_3} \in {\mathbb{R}^{M \times M}}$.
\begin{lemma}(Woodbury Matrix Identity \cite{higham2002accuracy}). \label{Lemma 1}
	If $Z \in {\mathbb{R}^{n \times n}}$ and $W\in{\mathbb{R}^{n \times n}}$ are invertible matrices, $U\in{\mathbb{R}^{n \times l}}$ and $V\in{\mathbb{R}^{l \times n}}$ are conformable matrices, then $(Z + U W V)^{-1} = Z^{-1} - Z^{-1} U (W^{-1} + V Z^{-1} U)^{-1} V Z^{-1}$.
\end{lemma}
\begin{lemma}(Pseudo-Inverse \cite{penrose_1955}). \label{Lemma 2}
	Given a constant matrix $B$, its pseudo-inverse matrix $B ^+$ has the following properties:\\
	1) $B^+BB^+=B^+$;\\
	2) $(B^TB)^+=B^+(B^T)^+$;\\
	3) $B^+=(B^TB)^+B^T$;\\
	4) $B^+=B^{-1}$ if $B$ is invertible.
\end{lemma}

\subsection{Problem Formulation}
Consider a heterogeneous MAS with one tracking leader, $N$ followers, and $M$ formation leaders. The tracking leader is responsible for providing a reference trajectory for motion, which can either represent an actual system or a virtual reference. The formation leaders have the task of achieving specific formations relative to the tracking leader, while the followers are required to position themselves within the convex hull formed by the formation leaders. The dynamics of each follower are given by 
\begin{equation}\label{1}
	{x_{i(k + 1)}} = {A_i}{x_{i(k)}} + {B_i}{u_{i(k)}},~~~ i \in \mathbb{F}
\end{equation}
and the dynamics of each formation leader are given by
\begin{equation}\label{2}
	x_{q(k + 1)}^l = {A_q}x_{q(k)}^l + B_q {u_{q(k)}^l},~~~ q \in \mathbb{L}  
\end{equation}
where $A_i \in {\mathbb{R}^{n \times n}}$ and ${B_i} \in {\mathbb{R}^{n \times {m_i}}}$. ${x_{i(k)}} \in {\mathbb{R}^n}$ and ${u_{i(k)}} \in {\mathbb{R}^{m_i}}$ are state and control input vector of follower $i$, respectively.  Moreover, $A_q \in {\mathbb{R}^{n \times n}}$ and ${B_q} \in {\mathbb{R}^{n \times {m_q}}}$. ${x_{q(k)}^l} \in {\mathbb{R}^n}$ and ${u_{q(k)}^l} \in {\mathbb{R}^{m_q}}$ are state and control input vector of formation leader $q$, respectively. Assume that the expected time-varying formation is specified by $h_{(k)}^l = col(h_{1(k)}^l,h_{2(k)}^l,...,h_{M(k)}^l)$, where $h_{q(k)}^l \in {\mathbb{R}^n}$ has the heterogeneous dynamics expressed as
\begin{equation}\label{3}
	h_{q(k+1)}^l = S_qh_{q(k)}^l.~~~ q \in \mathbb{L}
\end{equation}
Here, $S_q \in {\mathbb{R}^{n \times n}}$ determines the shape of formation. The dynamics of tracking leader are described as 
\begin{equation}\label{4}
	x_{(k + 1)}^o = A_0x_{(k)}^o,
\end{equation}
where $A_0 \in {\mathbb{R}^{n \times n}}$ and $x_{0(k)} \in {\mathbb{R}^n}$ is the state vector of tracking leader. 
\begin{remark}
	In the context of prior research on the FCC problem, the dynamics of agents are homogeneous \cite{dong2015formation,9409728,9870036}, i.e., ${x_{i(k + 1)}} = {A}{x_{i(k)}} + {B}{u_{i(k)}}, i \in \mathbb{F}\cup\mathbb{L}$, or the dynamics of agents are heterogeneous but the formation dynamics are homogeneous \cite{WANG201826,8384027,WANG2017392,8579593,GAO2021146}, i.e., $h_{q(k+1)}^l = Sh_{q(k)}^l,q \in \mathbb{L}$. To enhance the generality of the problem and broaden its applicability, we extend our study to a scenario where the systems exhibit full heterogeneity. 
\end{remark}

In traditional FCC and CC studies, the Laplacian matrix $\mathcal{L}_1^{-1}\mathcal{L}_2$ influences the tendency of followers to gravitate toward specific formation leaders (as observed in \cite{dong2015formation,9409728,9870036,WANG201826,8384027,WANG2017392,8579593,GAO2021146,HAGHSHENAS2015210,8277155,8869852}). However, given the disparities in dynamic and formation models, coupled with the unpredictable factors encountered in real-world applications, we anticipate that followers possess the ability to adjust their positions within the convex hull in response to signals emitted by leaders. For instance, in a UAV combat scenario, an outer leader serving as a protector may emit a distancing signal after coming under attack, with the intention of prompting the inner followers to move away from it and closer to the safer leaders. This signaling behavior is captured using the propensity factor ${{\vartheta _q}}>0$ (where $q \in \mathbb{L}$), which quantifies this signal. The higher the propensity factor of leader $q$, the more the followers will tend to converge toward leader $q$, as illustrated in Fig. \ref{fig1}. 
\begin{figure}
	\centering 
	\includegraphics[scale=1.3]{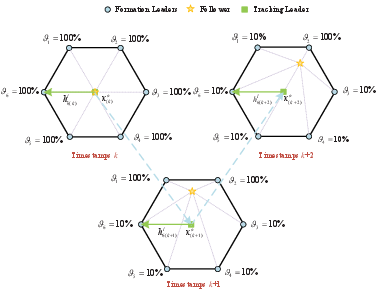}
	\caption {{ Example of the influence of propensity factors on follower convergence positions under different timestamps.}}
	\label{fig1}
\end{figure}

The assumptions and definitions for the PFCC problem are provided below.
\begin{assumption}\label{Assumption 1} 
	The directed network $\mathcal{G}$ contains a spanning tree with the tracking
	leader as the root. Additionally, for each follower, there exists at least one formation leader that has a directed path to it. 
\end{assumption}
\begin{assumption}\label{Assumption 2} 
	Every formation leader possesses a propensity factor, denoted as ${{\vartheta _q}}$, which is greater than zero (${{\vartheta _q}}>0$), and each formation leader is aware of their specific factor.
\end{assumption}

\begin{figure*}
	\centering 
	\includegraphics[scale=1.0]{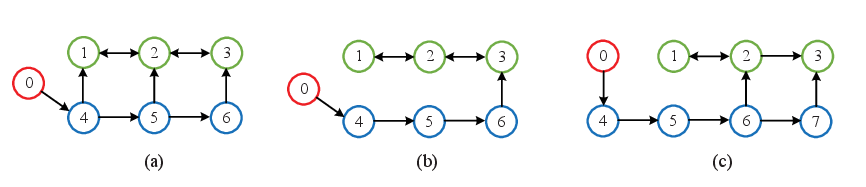}
	\caption {{ Examples to describe IFLs and ITFLs.}}
\end{figure*}

In Fig. \ref{fig1}, a single follower is illustrated for simplicity. However, in practice, there can be multiple followers, and the convergence position of each follower may not be influenced by all the leaders. Thus, it is essential to introduce the following two concepts.

\begin{definition}(Influential Formation Leaders).
	Define the set of influential formation leaders (IFLs) of follower $i$ as $\mathcal{N}_{i}^F= \{q\  \arrowvert \ \bar g_i^q > 0, q \in \mathbb{L}\}$, covering formation leaders who have a directed path to follower $i$. Define the set of influential formation leaders of leader $q$ as $\mathcal{N}_{q}^L= \{j\  \arrowvert \ \bar a_{qj}>0, j \in \mathbb{L}\}$, covering formation leaders who have a directed path to leader $q$. 
\end{definition}
\begin{definition}(Influential Transit Formation Leaders). 
	A formation leader is called an Influential Transit Formation Leader (ITFL) if it has at least one neighbor that acts as an Influential Formation Leader (IFL) for certain followers. However, this particular neighbor cannot directly relay information to these followers and relies on the ITFL to facilitate communication. 
\end{definition}

Using the graphs depicted in Fig. 2 as an example, where the red, blue, and green circles represent the tracking leader, formation leader, and follower, respectively.  In Fig. 2(a), the set of IFLs for all followers is $\{4, 5, 6\}$, and there are no ITFLs as no formation leader depends on other leaders to relay information. In Fig. 2(b), the set of IFLs for all followers remains $\{4, 5, 6\}$. However, as leader 4 relies on leaders 5 and 6, and leader 5 depends on leader 6 for information, the sets of ITFLs for node 4 and node 5 are $\{5, 6\}$ and $\{6\}$. In Fig. 2(c), the set of IFLs for followers 1 and 2 is $\{4,5,6\}$, while for follower 3, it is $\{4,5,6,7\}$. Formation leaders 4 and 5 have ITFL sets $\{5, 6, 7\}$ and $\{6, 7\}$, respectively, whereas leaders 6 and 7 do not have any ITFLs. 

 Current control methods discussed in the literature, including references \cite{HAGHSHENAS2015210,8277155,8869852,9627528,9802518,WANG201826,8384027,WANG2017392,8579593,GAO2021146,cheng2023data}, face limitations in effectively incorporating information from ITFLs. This results in an incomplete utilization of leader information under specific topological conditions. To illustrate this, consider Fig. 2(b), where all followers can only achieve tracking control for leader 6, neglecting containment control for all leaders. Combining the previous discussion, the PFCC problem is defined as follows.
 \begin{definition}(Propensity Containment Control). \label{Definition 2}
 	Consider systems \eqref{1}-\eqref{4}, the followers achieve propensity containment control if the following equation holds for $\forall i \in \mathbb{F}$:
 	\begin{equation}	
 		\mathop {\lim }\limits_{k \to \infty }  \Arrowvert{x_{i(k)}} - \sum\limits_{q = 1}^M \alpha_{i}^q x_{q(k)}^l\Arrowvert = 0,~~~ i \in \mathbb{F}
 	\end{equation}
 	where $\alpha_{i}^q={\frac{{\bar g_i^q{\vartheta _q}}}{{\sum\limits_{j = 1}^M {\bar g_i^j{\vartheta _j}} }}} \in \mathbb{R}$. It can be readily deduced that $\sum\limits_{q = 1}^M \alpha_{i}^q=1$.
 \end{definition}
 \begin{definition}(Propensity Formation-Containment Control). 
 	The heterogeneous systems \eqref{1} and \eqref{2} achieve propensity formation-containment control if all followers achieve propensity containment control and all formation leaders satisfy the following equation:
 	\begin{align}	
 		\mathop {\lim }\limits_{k \to \infty }\Arrowvert x_{q(k)}^l-h_{q(k)}^l-x_{(k)}^o\Arrowvert  = 0,~~~ q \in \mathbb{L}
 	\end{align}
 	for any initial states $x_{i(0)},i \in \mathbb{F}$ and $x_{q(0)},q \in \mathbb{L}$.	
 \end{definition}
 
The primary differences between the proposed PFCC and the traditional two-layer FCC are twofold: First, in the PFCC framework, the convergence weights of the followers are directly influenced by the propensity factors of their IFLs, rather than relying on the Laplacian matrix of the topological graph. This allows the formation leaders to adjust the followers' position without changing the graph structure. Second, the PFCC approach incorporates ITFLs, enabling followers to effectively leverage information from all their IFLs.
 
Moreover, in many existing studies (e.g., \cite{dong2015formation,9409728,9870036,WANG201826,8384027,WANG2017392,8579593,GAO2021146}), controller design relies heavily on explicit knowledge of the system's model. However, for certain multi-agent or large-scale systems where the model information is entirely unknown, these model-based approaches become ineffective. Consequently, it is valuable to explore controller design methods that utilize only data-driven information. The problem can be defined as follows.
 \begin{problem}\label{Problem 1}
 	Consider the MAS \eqref{1}-\eqref{4}. Designing control protocols for agents based on measurement data rather than dynamic models to ensure that followers efficiently utilise information from all formation leaders for accurate PFCC.
 \end{problem} 

In order to solve Problem 1, the following standard assump-
tions are made in this paper:
\begin{assumption}\label{Assumption 3} 
	$(A_i, B_i)$,  $(A_q, B_q)$, $A_0$ and $S_q$ are unknown, the pair $(A_i,B_i)$ and  $(A_q, B_q)$ are  stabilizable for $\forall i \in \mathbb{F}$ and $\forall q \in \mathbb{L}$.
\end{assumption}	
\begin{assumption}\label{Assumption 4} 
	All the eigenvalues of $S_q$, $\forall q \in \mathbb{L}$, and $A_0$ have modulus smaller than or equal to 1.
\end{assumption}
\begin{assumption}\label{Assumption 5} 
	The following linear matrix equations	
	\begin{numcases}{}
		{S_q} = {A_i} + {B_i}{U_{i,q}^h},~~~ i \in \mathbb{F}\label{5} \\ [2mm] 
		{S_q} = {A_q} + {B_q}{U_{q}^h},~~~ q \in \mathbb{L}\label{6}
	\end{numcases}
	have solutions ${U_{i,q}^h} \in {\mathbb{R}^{{m_i} \times n}}$ for $\forall i \in \mathbb{F}$ and ${U_{q}^h} \in {\mathbb{R}^{{m_q} \times n}}$ for $\forall q \in \mathbb{L}$.
\end{assumption}
\begin{assumption}\label{Assumption 6} 
	The following linear matrix equations
	\begin{numcases}{}
		{A_0} = {A_i} + {B_i}{U_{i}^o},~~~ i \in \mathbb{F}\label{7}\\[2mm] 
		{A_0} = {A_q} + {B_q}{U_{q}^o},~~~ q \in \mathbb{L}\label{8}
	\end{numcases}
	have solutions ${U_{i}^o} \in {\mathbb{R}^{{m_i} \times n}}$ for $\forall i \in \mathbb{F}$ and ${U_{q}^o} \in {\mathbb{R}^{{m_q} \times n}}$ for $\forall q \in \mathbb{L}$.
\end{assumption}
\begin{remark}
Assumption \ref{Assumption 4} is a standard assumption to ensure the  leader states do not diverge. Assumptions \ref{Assumption 5} and \ref{Assumption 6} align with standard practices found in references \cite{WANG201826,8384027,WANG2017392,8579593,GAO2021146,CAI2017299,6026912,DENG201962}, while \eqref{5}-\eqref{8} represent the state regulation equations. Importantly, in this paper, we do not necessitate knowledge of the system matrices or the need to resolve the regulation equations. We merely require the existence of equations \eqref{5}-\eqref{8}. 
\end{remark}

\section{Main Results}\label{S3}
\subsection{Model-Based Control Method for PFCC Problem}
In this section, we establish a distributed algorithm for the real-time collection of necessary control information and create a model-based control protocol to actualize PFCC. This work forms the theoretical basis for the ensuing data-driven approach.

Before formulating the controller, it is essential to establish the following definitions, which will have a pivotal role in the subsequent design and its verification.
\begin{definition}(Influential Propensity Factors).
	Define the set of influential propensity factors (IPFs) of follower $i$ as $\mathcal{P}_i^F=\{\vartheta _q\  \arrowvert \ q \in \mathcal{N}_{i}^F\}$, covering the propensity factors of influential leaders of follower $i$. Define the set of influential propensity factors of formation leader $q$ as $\mathcal{P}_q^L=\{\vartheta _j\  \arrowvert \ j \in \mathcal{N}_{q}^L\}$, covering the propensity factors of influential leaders of leader $q$.
\end{definition}
\begin{definition}(Dictionary Functions).
	Define the dictionary function $dic_i^F:\mathcal{N}_i^F \rightarrow \mathcal{P}_i^F$ of follower $i$ as $dic_i^F=\{(q,\vartheta _q)\ \arrowvert \ q \in \mathcal{N}_{i}^F\}$ with the domain of $\mathcal{N}_i^F$ and the range of $\mathcal{P}_i^F$, implies that $dic_i^F(q)=\vartheta _q$. Define the dictionary function $dic_q^L:\mathcal{N}_q^L \rightarrow \mathcal{P}_q^L$ of formation leader $q$ as $dic_q^L=\{(j,\vartheta _j)\ \arrowvert \ j \in \mathcal{N}_{q}^L\}$ with the domain of $\mathcal{N}_q^L$ and the range of $\mathcal{P}_q^L$.
\end{definition}

\begin{definition}(Real-time Augmented Systems).
	The real-time augmented systems of formation leaders and followers are defined respectively as 
	\begin{numcases}{}
		\hat X_{q(k + 1)}^l = {{\bar A}_q}\hat X_{q(k)}^l + {{\bar B}_q}u_{q(k)}^l,~~~~~~ q \in \mathbb{L} \\ [2mm]
		{\hat X_{i(k + 1)}} = {{\bar A}_{i(k)}}{\hat X_{i(k)}} + {{\bar B}_{i(k)}}{u_{i(k)}}, ~~i \in \mathbb{F}
	\end{numcases}
	with
	\begin{numcases}{}
		{{\bar A}_q} = diag(A_q,S_q,A_0),{{\bar B}_q} = \left[ {\begin{array}{*{20}{c}}
				{{B_q}}\\
				\bm{0}_{2n\times m_i}
		\end{array}} \right],~~~~ q \in \mathbb{L} \nonumber \\ [2mm]
		{{\bar A}_{i(k)}} = diag({A_i},{S_{{\phi _{i(k)}}(1)}},...,{S_{{\phi _{i(k)}}({{\mathcal I}_{i(k)}})}},{A_0}),\nonumber\\
		~~~~~~~~~~~~~~~~~~{{\bar B}_{i(k)}} = \left[ {\begin{array}{*{20}{c}}
				{{B_{i}}}\\
				{{\bm{0}_{(1 + {{\mathcal I}_{i(k)}})n \times {m_i}}}}
		\end{array}} \right], ~~~~i \in \mathbb{F}\nonumber
	\end{numcases}
	where $\hat X_{q(k)}^l=col(x_{q(k)}^l, h_{q(k)}^l,\hat x_{q(k)}^o)\in \mathbb{R}^{3n}$ and $\hat X_{i(k)}=col(x_{i(k)}, \hat h_{i(k)}^{{\phi _{i(k)}}(1)},...,\hat h_{i(k)}^{{\phi _{i(k)}}({{\mathcal I}_{i(k)}})},\hat x_{q(k)}^o)\in \mathbb{R}^{(2+{{\mathcal I}_{i(k)}})n}$ are the augmented system
	state vectors of leader $q$ and follower $i$, in which ${{\mathcal I}_{i(k)}}=|\mathcal{N}_{i(k)}^F|$ is the number of elements of ${\mathcal N}_{i(k)}^F$ and ${\phi _{i(k)}}$ is the sequential arrangement of ${\mathcal N}_{i(k)}^F$.
\end{definition} 
\begin{definition}(Stabilizing Feedback Gains).
	Define the set of stabilizing feedback gains of $(A_i,B_i),i \in \mathbb{F}$ as $\mathcal{K}_i=\{K_i\ |\ \rho(A_i+B_iK_i)<1\}$. Define the set of stabilizing feedback gains of $(A_q,B_q),q \in \mathbb{L}$ as $\mathcal{K}_q=\{K_q\ |\ \rho(A_q+B_qK_q)<1\}$.
\end{definition}

To successfully implement PFCC, the sets of IFLs,  $\mathcal{N}_{i}^F$, $\mathcal{N}_{q}^L$,  dictionary function $dic_i^F$, convex coefficient $\alpha_{i}^q$ and weights $\bar g_i^q$, $\bar a_{qm(k)}$ are crucial for each agent. We have developed Algorithm \ref{Algorithm 1} to gather this essential information in real time through a distributed process.
\begin{algorithm}[htbp]	\footnotesize
	\renewcommand{\algorithmicrequire}{Initialization:}
	\caption{Distributed algorithm to obtain ILs, IPFs, convex coefficients and weights in real-time}\label{Algorithm 1}
	\begin{algorithmic}[1]
		\State Set $\mathcal{N}_{i(0)}^F=\{q \ \arrowvert \  g_i^q > 0, q \in L \}$, $\mathcal{N}_{q(0)}^L=\{j \ \arrowvert \  j \in \mathcal{N}_q \}$, $dic_{i(0)}=\{(q,\vartheta _q)\ \arrowvert \ q \in \mathcal{N}_{i(0)}^F\}$ and $dic_{q(0)}=\{(j,\vartheta _j)\ \arrowvert \ j \in \mathcal{N}_{q(0)}^L\}$. Let $k=0$.
		\State Set $\bar g_{i(k)}^q = 
		\left\{
		\begin{array}{ll}
			1,~~~~~~~~~~~~~~~~~~~~~~~~~~~~~q\in \mathcal{N}_{i(k)}^F\\[2mm]
			0,~~~~~~~~~~~~~~~~~~~~~~~~~~~~~q \notin \mathcal{N}_{i(k)}^F
		\end{array}\right.$  $\forall i \in \mathbb{F}$.
		
		\State Set $\alpha_{i(k)}^q = 
		\left\{
		\begin{array}{ll}
			{\frac{dic_{i(k)}^F(q)}{\sum\limits_{j = 1}^M \bar g_{i(k)}^jdic_{i(k)}^F(j)}} ,~~~~~~~~~q\in \mathcal{N}_{i(k)}^F\\[2mm]
			0,~~~~~~~~~~~~~~~~~~~~~~~~~~~~~q \notin \mathcal{N}_{i(k)}^F
		\end{array}\right.$  $\forall i \in \mathbb{F}$.
		
		\State Set $\bar a_{qm(k)} = 
		\left\{
		\begin{array}{ll}
			1,~~~~~~~~~~~~~~~~~~~~~~~~~m\in \mathcal{N}_{q(k)}^L\\[2mm]
			0,~~~~~~~~~~~~~~~~~~~~~~~~~m \notin \mathcal{N}_{q(k)}^L
		\end{array}\right.$  $\forall q \in \mathbb{L}$.
		\State Calculate the set of ITFLs at time $k+1$:
		\begin{equation*}
			\mathcal{N}_{i(k+1)}^F = \mathcal{N}_{i(k)}^F\bigcup_{i\neq j=M+1 }^{M+N} a_{ij}\mathcal{N}_{j(k)}^F\bigcup_{q=1 }^Mg_i^q\mathcal{N}_{q(k)}^L,~~~ i \in \mathbb{F}
		\end{equation*}
		\begin{equation*}
			\mathcal{N}_{q(k+1)}^L = \mathcal{N}_{q(k)}^L\bigcup_{q\neq j=1 }^{M} a_{qj}\mathcal{N}_{j(k)}^L,~~~ q \in \mathbb{L}.
		\end{equation*}
		\State Calculate the set of IPFs at time $k+1$:
		\begin{equation*}
			dic_{i(k+1)}^F = dic_{i(k)}^F\bigcup_{i\neq j=M+1 }^{M+N} a_{ij}dic_{j(k)}^F\bigcup_{q=1 }^Mg_i^qdic_{q(k)}^L, ~i \in \mathbb{F}
		\end{equation*}
		\begin{equation*}
			dic_{q(k+1)}^L = dic_{q(k)}^L\bigcup_{q\neq j=1 }^{M} a_{qj}dic_{j(k)}^L,~~~ q \in \mathbb{L}.
		\end{equation*}
		\State Let $k=k+1$ and return to step 2.
	\end{algorithmic}
\end{algorithm}
\begin{proposition}\label{Proposition 1}
	Consider a directed network topology $\mathcal{G}$ composed of one tracking leader, $M$ formation leaders, and $N$ followers, satisfying Assumption \ref{Assumption 1}. Then by employing Algorithm \ref{Algorithm 1}, after a maximum of $N+M-1$ iterations, (1) $\alpha_{i(k)}^q \rightarrow \alpha_{i}^q$, $\bar g_{i(k)}^q\rightarrow\bar g_{i}^q$, $q\in \mathcal{N}_{i}^F$, $\forall i \in \mathbb{F}$, (2) $\bar a_{qm(k)}\rightarrow\bar a_{qm}$, $m\in \mathcal{N}_{q}^L$, $\forall q \in \mathbb{L}$, (3) $\mathcal{N}_{i(k)}^F \rightarrow\mathcal{N}_{i}^F$, $\forall i \in \mathbb{F}$, (4) $\mathcal{N}_{q(k)}^L \rightarrow\mathcal{N}_{q}^L$, $\forall q \in \mathbb{L}$.
\end{proposition}
\begin{proof}
	 According to the definitions of $\alpha_{i}^q$, $\bar g_{i}^q$, $\bar a_{qm}$, $\mathcal{N}_{i}^F$ and  $\mathcal{N}_{q}^L$, when employing distributed Algorithm 1, the topology structure that results in the longest time for all agents to acquire the necessary information is a fully chained structure. It takes $N+M-1$ iterations to transmit information from the top agent of the chain structure to the last agent. Hence, the maximum number of iterations required to obtain the required information through Algorithm 1 will not exceed $N+M-1$. The proof is completed. 
\end{proof}

Then, the following model-based state feedback controllers are designed to solve PFCC problem:
\begin{numcases}{}
	u_{q(k)}^l = {K_{q1}}x_{q(k)}^l + {K_{qh}}h_{q(k)}^l + {K_{qo}}\hat x_{q(k)}^o, \label{17} \\[2mm] 
	{u_{i(k)}} = K_{i1}{x_{i(k)}} + {K_{io}}\hat x_{i(k)}^o+ \sum\limits_{q \in {\mathcal N}_{i(k)}^F} {K_{ih}^q\alpha _{i(k)}^q\hat h_{i(k)}^{q}},\label{18}
\end{numcases}
where $\hat x_{q(k)}^o\in \mathbb{R}^n $, $\hat x_{i(k)}^o\in \mathbb{R}^n$ and $\hat h_{i(k)}^{q}\in \mathbb{R}^n$ represent observations made by leader $q$ for $x_{(k)}^o$, and by follower $i$ for $x_{(k)}^o$ and $h_{q(k)}^l$, respectively, which are designed in the later. $K_{q1}\in \mathbb{R}^{m_q\times n},K_{qh}\in \mathbb{R}^{m_q\times n},K_{qo}\in \mathbb{R}^{m_q\times n}$ are the control gain matrices of leader $q$ and $K_{i1}\in \mathbb{R}^{m_i\times n},K_{ih}^q\in \mathbb{R}^{m_i\times n},K_{io}\in \mathbb{R}^{m_i\times n}$ are the control gain matrices of follower $i$.
\begin{lemma} \label{Lemma 3}
	For any solvable matrix equation of the form $S=A+BM$ where $M$ is the matrix to be solved, the equation
	\begin{equation}
		B^+B\bar M=\bar M \nonumber
	\end{equation}
	remains true, where $\bar M$ is the minimum norm solution $\bar M=B^+ (S - A)$.
\end{lemma}
\begin{proof}
	 Since $\bar M=B^+(S-A)$, we have $B^+B\bar M=B^+BB^+(S-A)=B^+(S-A)=\bar M$, where the second equal sign is obtained from Lemma \ref{Lemma 2}. The proof is completed.
\end{proof}

The next result reveals the relationship between the regulation equations and the designed control gains.
\begin{lemma} \label{Lemma 4}
	Given the control protocols $\eqref{17}$ and $\eqref{18}$, design the control gain matrices $K_q^L = \left[ {{K_{q1}},{K_{qh}},{K_{qo}}} \right], q \in \mathbb{L}$ and $K_{i(k)}^F = \left[ {{K_{i1}},K_{ih}^{{\phi _{i(k)}}(1)},...,K_{ih}^{{\phi _{i(k)}}({{\mathcal I}_{i(k)}})},{K_{io}}} \right], i \in \mathbb{F}$ as
	\begin{numcases}{}
		K_q^L =  - {(\bar B_q^T{P_q}{{\bar B}_q})^{ +}}\bar B_q^T{P_q}{{\bar A}_q},~~~~~~~~~~~~~~~~~~~~ \label{19}\\ [2mm]
		K_{i(k)}^F =  - {(\bar B_{i(k)}^T{P_{i(k)}}{{\bar B}_{i(k)}})^{ +}}\bar B_{i(k)}^T{P_{i(k)}}{{\bar A}_{i(k)}},  \label{20}
	\end{numcases}
	$P_q,q \in \mathbb{L}$ and $P_{i(k)},i \in \mathbb{F}$  satisfies the following Riccati equations:
	\begin{numcases}{}
		{P_q} = C_l^T{Q_q}{C_l} +{({{\bar A}_q} + {{\bar B}_q}K_q^L)^T}{P_q}({{\bar A}_q} + {{\bar B}_q}K_q^L), \label{21}\\ [2mm]
		{P_{i(k)}} = C_{i(k)}^T{Q_i}{C_{i(k)}}	+  {({{\bar A}_{i(k)}} + {{\bar B}_{i(k)}}K_{i(k)}^F)^T} \nonumber\\
		~~~~~~~~~~~~~~~~~~~~~~~~~~~~~~\times{P_{i(k)}}({{\bar A}_{i(k)}} + {{\bar B}_{i(k)}}K_{i(k)}^F), \label{22}
	\end{numcases}
	where ${C_i} = [{\bm{I}_n}, - \alpha _{i(k)}^{{\phi _{i(k)}}(1)}{\bm{I}_n},..., - \alpha _{i(k)}^{{\phi _{i(k)}}({{\mathcal I}_{i(k)}})}{\bm{I}_n},-{\bm{I}_n}]$, ${C_l} = \left[ {{\bm{I}_n}, - {\bm{I}_n}, - {\bm{I}_n}} \right]$, $Q_q=Q_q^T>0$, $Q_i=Q_i^T>0$. Then, the control gain matrices satisfy
	\begin{align}{}
		&\left\{
		\begin{array}{ll}
			K_{q1}\in \mathcal{K}_q, \\[2mm]
			K_{q1} + K_{qh} = \bar U_q^h,~~ K_{q1} + K_{qo} = \bar U_q^o, ~~~ q \in \mathbb{L}
		\end{array}\right.\label{23}  \\
		&\left\{
		\begin{array}{ll}
			K_{i1}\in\mathcal{K}_i,\\
			K_{i1} +\dfrac{1}{\alpha _{i(k)}^{{\phi _{i(k)}}(j)}} K_{ih}^{{\phi _{i(k)}}(j)} = \bar U_{i,{{\phi _{i(k)}}(j)}}^h, \\[2mm]
			K_{i1} + K_{io} = \bar U_i^o,~~~~~~~j=1,\cdots,{{\mathcal I}_{i(k)}},~i \in \mathbb{F},
		\end{array}\right. \label{24}
	\end{align}
	where $\bar U_{i,{{\phi _{i(k)}}(j)}}^h$, $\bar U_q^h$, $\bar U_i^o$ and $\bar U_q^o$ are the minimum norm solutions of \eqref{5}-\eqref{8}, respectively.
\end{lemma}
\begin{proof}
	 To begin with, demonstrate the accuracy of \eqref{24}. According to \eqref{22}, $P_{i(k)}=P_{i(k)}^T > 0$ can be obtained, without loss of generality, divide $P_{i(k)}$ into $P_{i(k)}=\left[ {\begin{array}{*{20}{c}}
			{P_{i(k)}^{1}}\\
			\vdots \\
			{P_{i(k)}^{2+{{\mathcal I}_{i(k)}}}}
	\end{array}} \right]$, where $P_{i(k)}^{m}= \left[{\begin{array}{*{20}{c}}
			{P_{i(k)}^{m,1}}& \cdots &{P_{i(k)}^{m,2+{{\mathcal I}_{i(k)}}}}
	\end{array}}\right] \in\mathbb{R}^{n \times (2+{{\mathcal I}_{i(k)}})n} ,m =1,...,2+{{\mathcal I}_{i(k)}}$. Then, expanding \eqref{20} yields
	\begin{align}
		&K_{i(k)}^F= -(B_i^T{P_{i(k)}^{1,1}}B_i)^{+}\nonumber \\ &~~~~~~~~~~~\times\left[B_i^TP_{i(k)}^{1,1}\  B_i^TP_{i(k)}^{1,2}\cdots B_i^TP_{i(k)}^{1,2+{{\mathcal I}_{i(k)}}}\right]  \nonumber \\
		&~~~~~~~~~~~~\times \left[ {\begin{array}{*{20}{c}}
				{{A_i}}&{}&{}&{}\\
				{}&{{{S_{{\phi _{i(k)}}(1)}}}}&{}&{}\\
				{}&{}& \ddots &{}\\
				{}&{}&{}&{{A_0}}
		\end{array}} \right].
	\end{align}
	Then, we can obtain the following equations for $j=1,\cdots,{{\mathcal I}_{i(k)}},~i \in \mathbb{F}$:
	\begin{numcases}{}
		K_{i1} +\dfrac{1}{\alpha _{i(k)}^{{\phi _{i(k)}}(j)}} K_{ih}^{{\phi _{i(k)}}(j)} = -(B_i^T{P_{i(k)}^{1,1}}B_i)^{+}B_i^T \nonumber\\\label{51}
		\hphantom{K_{i1} +}~~~~\times(P_{i(k)}^{1,1}A_i+\dfrac{1}{\alpha _{i(k)}^{{\phi _{i(k)}}(j)}}P_{i(k)}^{1,1+j}S_{{\phi _{i(k)}}(j)}),  \\[2mm]
		K_{i1} + K_{io} =  -(B_i^T{P_{i(k)}^{1,1}}B_i)^{+}B_i^T\nonumber\\\label{52}
		\hphantom{K_{i1} + K_{io}=  -}~~~~~ \times(P_{i(k)}^{1,1}A_i+P_{i(k)}^{1,2+{{\mathcal I}_{i(k)}}}A_0).
	\end{numcases}
	
	Next, recalling \eqref{22}, simultaneously multiply $diag({\bm{I}_n}, \dfrac{1}{\alpha _{i(k)}^{{\phi _{i(k)}}(1)}}{\bm{I}_n},..., \dfrac{1}{\alpha _{i(k)}^{{\phi _{i(k)}}({{\mathcal I}_{i(k)}})}}{\bm{I}_n} ,{I_n})$ on the right-hand side of $P_{i(k)}$, the first group of transformed block matrices is
	\begin{align}\label{53}
		&{}
		\begin{bmatrix} 
			{P_{i(k)}^{1,1}}\ \dfrac{1}{\alpha _{i(k)}^{{\phi _{i(k)}}(1)}}P_{i(k)}^{1,2}\  \cdots \ {P_{i(k)}^{1,2+{{\mathcal I}_{i(k)}}}}
		\end{bmatrix}=\begin{bmatrix} 
			Q_i\ -Q_i\cdots-Q_i
		\end{bmatrix}\nonumber \\
		&+{({A_i} + {B_i}{K_{i1}})}^T[{P_{i(k)}^{1,1}}\ {P_{i(k)}^{1,2}}\   \cdots \ {P_{i(k)}^{1,2+{{\mathcal I}_{i(k)}}}}]\nonumber \\
		&\times \scriptsize \left[ {\begin{array}{*{20}{c}}
				{{A_i+B_iK_{i1}}}&{\dfrac{1}{\alpha _{i(k)}^{{\phi _{i(k)}}(1)}}B_i K_{ih}^{{\phi _{i(k)}}(1)}}&\cdots&{B_iK_{io}}\\
				{}&{\dfrac{1}{\alpha _{i(k)}^{{\phi _{i(k)}}(1)}}{{S_{{\phi _{i(k)}}(1)}}}}&{}&{}\\
				{}&{}& \ddots &{}\\
				{}&{}&{}&{{A_0}}
		\end{array}} \right]&
	\end{align}
	and the following equations can be obtained for $j=1,\cdots,{{\mathcal I}_{i(k)}},~i \in \mathbb{F}$:
	\begin{equation}\label{54} 	
		\left\{
		\begin{array}{ll}
			P_{i(k)}^{1,1} +\dfrac{1}{\alpha _{i(k)}^{{\phi _{i(k)}}(j)}} P_{i(k)}^{1,1+j} = {({A_i} + {B_i}{K_{i1}})}^T\dfrac{1}{\alpha _{i(k)}^{{\phi _{i(k)}}(j)}} \\
			~~~~~~~~~~~~~~\times P_{i(k)}^{1,1+j}S_{{\phi _{i(k)}}(j)}	+{({A_i} + {B_i}{K_{i1}})}^TP_{i(k)}^{1,1} \\
			~~~~~~~~~~~~~~~~~\times(A_i+B_iK_{1i}+B_i\dfrac{1}{\alpha _{i(k)}^{{\phi _{i(k)}}(j)}} K_{ih}^{{\phi _{i(k)}}(j)}), \\[2mm]
			P_{i(k)}^{1,1} + P_{i(k)}^{1,2+{\mathcal I}_{i(k)}} = {({A_i} + {B_i}{K_{i1}})}^T P_{i(k)}^{1,2+{\mathcal I}_{i(k)}}A_0 \\
			~~~~~~~~+{({A_i} + {B_i}{K_{i1}})}^TP_{i(k)}^{1,1}(A_i+B_iK_{1i}+B_i K_{io}).
		\end{array}\right.
	\end{equation}
	Define $\mathcal{X}_{i(k)}^{{\phi _{i(k)}}(j)}= K_{i1} +\dfrac{1}{\alpha _{i(k)}^{{\phi _{i(k)}}(j)}} K_{ih}^{{\phi _{i(k)}}(j)}$, $\mathcal{X}_{i}^o=	K_{i1} + K_{io}$, $\mathcal{Y}_{i(k)}^{{\phi _{i(k)}}(j)}= 	P_{i(k)}^{1,1} +\dfrac{1}{\alpha _{i(k)}^{{\phi _{i(k)}}(j)}} P_{i(k)}^{1,1+j}$ and $\mathcal{Y}_{i}^o=P_{i(k)}^{1,1} + P_{i(k)}^{1,2+{\mathcal I}_{i(k)}}$, then from \eqref{5} and \eqref{51}, we have 
	\begin{flalign}\label{55}
		&\mathcal{X}_{i(k)}^{{\phi _{i(k)}}(j)} = -(B_i^T{P_{i(k)}^{1,1}}B_i)^{+}B_i^T\mathcal{Y}_{i(k)}^{{\phi _{i(k)}}(j)}A_i\nonumber \\
		&~~~~-(B_i^T{P_{i(k)}^{1,1}}B_i)^{+}B_i^T\dfrac{1}{\alpha _{i(k)}^{{\phi _{i(k)}}(j)}} P_{i(k)}^{1,1+j}B_i\bar U_{i,{{\phi _{i(k)}}(j)}}^h\nonumber \\
		&=-(B_i^T{P_{i(k)}^{1,1}}B_i)^{+}B_i^T\mathcal{Y}_{i(k)}^{{\phi _{i(k)}}(j)}S_{{\phi _{i(k)}}(j)}+B_i^+B_i\bar U_{i,{{\phi _{i(k)}}(j)}}^h\nonumber \\
		&=-(B_i^T{P_{i(k)}^{1,1}}B_i)^{+}B_i^T\mathcal{Y}_{i(k)}^{{\phi _{i(k)}}(j)}S_{{\phi _{i(k)}}(j)}+\bar U_{i,{{\phi _{i(k)}}(j)}}^h,&
	\end{flalign}
	where the the second equal sign is obtained from Lemma \ref{Lemma 2} and the third equal sign is obtained from Lemma \ref{Lemma 3}. From \eqref{5} and \eqref{54}, we have
	\begin{flalign}\label{56}
		&\mathcal{Y}_{i(k)}^{{\phi _{i(k)}}(j)} =   {({A_i} + {B_i}{K_{i1}})}^T \mathcal{Y}_{i(k)}^{{\phi _{i(k)}}(j)}S_{{\phi _{i(k)}}(j)}\nonumber\\
		&+ {({A_i} + {B_i}{K_{i1}})}^T{P_{i(k)}^{1,1}}B_i(\mathcal{X}_{i(k)}^{{\phi _{i(k)}}(j)}-\bar U_{i,{{\phi _{i(k)}}(j)}}^h).&
	\end{flalign}
	Substituting \eqref{55} into \eqref{56} yields
	\begin{flalign}\label{57}
		\mathcal{Y}_{i(k)}^{{\phi _{i(k)}}(j)} =  &{({A_i} + {B_i}{K_{i1}})}^T (\bm{I}_n-{P_{i(k)}^{1,1}}B_i(B_i^T{P_{i(k)}^{1,1}}B_i)^{+}B_i^T)\nonumber \\
		&\times \mathcal{Y}_{i(k)}^{{\phi _{i(k)}}(j)}S_{{\phi _{i(k)}}(j)}.
	\end{flalign}

	According to \eqref{53}, we have $P_{i(k)}^{1,1}=Q_i+(A_i+B_iK_{i1})^TP_{i(k)}^{1,1}(A_i+B_iK_{i1})$, since $Q_i > 0$, it follows that the eigenvalues of $(A_i+B_iK_{i1})$ must be within the unit circle, i.e. $K_{i1}\in\mathcal{K}_i$. Therefore, the eigenvalues of matrices $S_{{\phi _{i(k)}}(j)}$ less than or equal to $1$, the eigenvalues of ${({A_i} + {B_i}{K_{i1}})}^T$ and $\bm{I}_n-{P_{i(k)}^{1,1}}B_i(B_i^T{P_{i(k)}^{1,1}}B_i)^{+}B_i^T$ less than $1$. Consequently, we can conclude that matrix $\mathcal{Y}_{i(k)}^{{\phi _{i(k)}}(j)}=\bm{0}_{n\times n}$. Substituting it into \eqref{55} to get $K_{i1} +\dfrac{1}{\alpha _{i(k)}^{{\phi _{i(k)}}(j)}} K_{ih}^{{\phi _{i(k)}}(j)} =\bar U_{i,{{\phi _{i(k)}}(j)}}^h$ for $j=1,\cdots,{{\mathcal I}_{i(k)}},i \in \mathbb{F}$. It is possible to use a similar method to conclude that $K_{i1} + K_{io} = \bar U_i^o,i \in \mathbb{F}$. So far, the proof of \eqref{24} has been completed. In addition, \eqref{19} and \eqref{21} can be considered as a special case of \eqref{20} and \eqref{22}, so the proof of \eqref{23} is easy to be obtained and omitted.
\end{proof}
\begin{theorem}\label{Theorem 1}
	Consider the heterogeneous MAS \eqref{1}-\eqref{4} with Assumptions \ref{Assumption 1}-\ref{Assumption 6} hold. The Problem \ref{Problem 1} can be solved using the model-based feedback controllers \eqref{17}-\eqref{18} with the control gains constructed by Lemma \ref{Lemma 4}, if the observations satisfy $\mathop {\lim }\limits_{k \to \infty } {\hat x_{q(k)}^o-x_{(k)}^o}= \bm{0}, q \in \mathbb{L}$, $\mathop {\lim }\limits_{k \to \infty } {\hat x_{i(k)}^o-x_{(k)}^o}= \bm{0}, i \in \mathbb{F}$ and $\mathop {\lim }\limits_{k \to \infty } {\hat h_{i(k)}^{q}-h_{q(k)}^{l}}= \bm{0}, q\in {\mathcal N}_{i}^F, i \in \mathbb{F}$.
\end{theorem}
\begin{proof}
  Define the formation error of leader $q$ as $e_{q(k)}^l = x_{q(k)}^l - h_{q(k)}^l - x_{(k)}^o, q \in \mathbb{L}$, according to \eqref{17}, we have
\begin{align}\label{58}
	e_{q(k + 1)}^l& = {A_q}x_{q(k)}^l + {B_q}u_{q(k)}^l - {S_q}h_{q(k)}^l - {A_0}x_{(k)}^o	\nonumber \\
	&=({A_q} + {B_q}K_{q1})x_{q(k)}^l - ({S_q} - {B_q}K_{qh})h_{q(k)}^l \nonumber \\
	&~~~~- ({A_0} - {B_q}K_{qo})x_{(k)}^o + {B_q}K_{qo}\tilde x_{q(k)}^o,&
\end{align}
where $\tilde x_{q(k)}^o = {\hat x_{q(k)}^o-x_{(k)}^o}$ represents the estimation error of leader $q$ to virtual node. Since that $K_{q1}, K_{qh}$ and $K_{qo}$ satisfy \eqref{23}, after simplification, one has
\begin{align}
	e_{q(k + 1)}^l=({A_q} + {B_q}K_{q1})e_{q(k)}^l + {B_q}K_{qo}\tilde x_{q(k)}^o.
\end{align}

As indicated in  Lemma \ref{Lemma 2}, $K_{q1}\in \mathcal{K}_q$, so $\mathop {\lim }\limits_{k \to \infty } {e_{qk}^l}= \bm{0}$ if $\mathop {\lim }\limits_{k \to \infty } {\hat x_{q(k)}^o-x_{(k)}^o}= \bm{0}, q \in \mathbb{L}$.

Define the containment error of follower $i$ as $e_{i(k)}^f = x_{i(k)} - \sum\limits_{q = 1}^M {\alpha _i^q} (h_{q(k)}^l + x_{(k)}^o), q \in \mathbb{L}$, according to \eqref{18}, we have
\begin{align}\label{60}
	&e_{i(k + 1)}^f= {A_i}{x_{i(k)}} + {B_i}{u_{i(k)}} - \sum\limits_{q = 1}^M {\alpha _i^q} ({S_q}h_{q(k)}^l + {A_0}x_{(k)}^o)\nonumber\\
	&~~= ({A_i} + {B_i}{K_{i1}}){x_{i(k)}} - ({A_0} - {B_i}{K_{io}})x_{(k)}^o \nonumber\\
	&~~~+ {B_i}\sum\limits_{q \in {\mathcal{N}}_{i(k)}^F} {K_{ih}^q\alpha _i^qh_{q(k)}^l}  - \sum\limits_{q = 1}^M {\alpha _i^q} {S_q}h_{q(k)}^l + {B_i}K_{io}\tilde x_{i(k)}^o \nonumber\\
	&~~~+ \sum\limits_{q \in {\mathcal{N}}_{i(k)}^F} {{B_i}K_{ih}^q(\alpha _{i(k)}^q\tilde h_{i(k)}^q}  + \tilde \alpha _{i(k)}^qh_{i(k)}^q),&
\end{align}
where $\tilde x_{i(k)}^o = {\hat x_{i(k)}^o-x_{(k)}^o}$ and $\tilde h_{i(k)}^q = \hat h_{i(k)}^q- h_{q(k)}^l$ represent the estimation errors of follower $i$ to virtual node and formation vector of leader $q$, respectively, $\tilde \alpha _{i(k)}^q=\hat \alpha _{i(k)}^q- \alpha _{i(k)}^q$. According to Definition \ref{Definition 2}, we know that $\alpha _{i}^q=0$ if $q \notin {\mathcal{N}}_{i}^F$. From Theorem \ref{Theorem 1}, ${\mathcal{N}}_{i(k)}^F={\mathcal{N}}_{i}^F$ and $\alpha _{i(k)}^q= \alpha _{i(k)}^q$ after sufficient steps. Therefore, substituting \eqref{24} into \eqref{60} yields 
\begin{align}
	e_{i(k + 1)}^f=& ({A_i} + {B_i}K_{i1})e_{i(k)}^f + {B_i}K_{io}\tilde x_{i(k)}^o \nonumber \\
	&+ \sum\limits_{q \in N_{i(k)}^F} {{B_i}K_{ih}^q\alpha _{i(k)}^q\tilde h_{i(k)}^q}  .
\end{align}

As indicated in the Lemma \ref{Lemma 2}, $K_{i1}\in \mathcal{K}_i$, so $\mathop {\lim }\limits_{k \to \infty } {e_{ik}^f}= \bm{0}$ if  $\mathop {\lim }\limits_{k \to \infty } {\hat x_{i(k)}^o-x_{(k)}^o}= \bm{0}, i \in \mathbb{F}$ and $\mathop {\lim }\limits_{k \to \infty } {\hat h_{i(k)}^{q}-h_{q(k)}^{l}}= \bm{0}, q\in {\mathcal N}_{i}^F, i \in \mathbb{F}$. The proof is completed.
\end{proof}

\begin{remark}
	The necessary condition for the regulation equation \eqref{5}-\eqref{8} to have a unique solution is that $B_i$ and $B_q$ are column full rank, indicating that the controlled system is either underactuated or fully actuated. In this case, we have ${(\bar B_m^T{P_m}{{\bar B}_m})^{ +}}={(\bar B_m^T{P_m}{{\bar B}_m})^{ -1}}, m \in \mathbb{F}\cup\mathbb{L}$ and the control gain can be replaced by the optimal control gain $- {(\bar B_m^T{P_m}{{\bar B}_m})^{ -1}}\bar B_m^T{P_m}{{\bar A}_m}$. However, for overactuated systems like the underwater vehicle system mentioned in \cite{TANAKITKORN201767}, the regulation equation has infinite solutions if any, and $\bar B_m^T{P_m}{{\bar B}_m}$ is non-invertible. Lemma \ref{Lemma 4} reveals the relationship between the designed pseudo-inverse controller and the minimum norm solution of the regulation equation, which guarantees the asymptotic convergence of the agents to the expected states. Thus, while the state PFCC problem under consideration necessitates an equal number of rows for $A_i$, $B_i$, $A_q$, $B_q$, $S_q$ and $A_0$, there are no additional requirements regarding the number of columns for $B_i$ and $B_q$.
\end{remark}

In this subsection, we have introduced a model-based control scheme for the PFCC problem. However, two key issues remain unresolved: 1) the observations used in controllers \eqref{17} and \eqref{18} are unknown, and 2) the development of a data-driven PFCC scheme. These two challenges will be addressed in Section III. B and Section III. C, respectively.

\subsection{Data-Based Distributed Adaptive Observer}
In this subsection, we develop the data-based adaptive observers for both leaders and followers to obtain the values required in controllers \eqref{17} and \eqref{18}. To remove the dependency on a system model within the observer, we integrate the concept of recursive least squares into the conventional observer framework, allowing for simultaneous estimation of both the target state and the target model matrix. The effectiveness of the proposed observers is demonstrated through Theorem \ref{Theorem 2} and Theorem \ref{Theorem 3}.

\textbf{Observer for tracking leader by formation leaders.} For formation leader $q$, the following observer is designed to estimate $x_{(k)}^o$:
\begin{numcases}{}
	{L_{q(k + 1)}} = {L_{q(k)}} - {L_{q(k)}}({\bar x _{q(k)}^o})^T\nonumber\\
	~~~~~~~~~~~~\times{(\bm{I}_{n^2} + {\bar x _{q(k)}^o}{L_{q(k)}} ({\bar x _{q(k)}^o})^T)^{ - 1}} {\bar x _{q(k)}^o}{L_{q(k)}},\label{25}\\[2mm]
	\hat A_{q,0(k + 1)}^{vec}  \nonumber\\
	~~~~~~~~=\hat A_{q,0(k)}^{vec}- {\Lambda}_l{\bar x _{q(k)}^o}{(L_{q(k+1)}^{ - 1} + \xi \bm{I}_n)^{ - 1}}{\eta _{q(k + 1)}}, \\[2mm]
	{\hat x _{q(k + 1)}^o} = \hat A_{q,0(k)}{\hat x _{q(k)}^o} - \mu_lF_l {\eta _{q(k)}},\label{27}
\end{numcases}
where ${\bar x _{q(k)}^o}={\bm{I}_n} \otimes {\hat x _{q(k)}^o}$, ${\eta _{q(k)}} = \sum\limits_{j=N+1}^{N+M} {{a_{qj}}({\hat x _{q(k)}^o} - {\hat x _{j(k)}^o})}  + {g_{q}^0({\hat x _{q(k)}^o} - { x _{(k)}^o})} $, $\xi\geq 1$, ${\Lambda}_l>0$ and $\mu_l > 0$ are constant values, $F_l$ is constant matrix.

\begin{lemma} \label{Lemma 5}
	Let the iterative formula of the adaptive parameter matrix ${L_{q(k)}}$ be given by \eqref{25} and the initial value ${L_{q(0)}}>0$. Then,
	\begin{equation}
		{L_{q(k + 1)}^{-1}}>({\bar x _{q(k)}^o})^T{\bar x _{q(k)}^o} \nonumber
	\end{equation}
	and 
	\begin{equation}
		({\bar x _{q(k)}^o})^T{\bar x _{q(k)}^o}{(L_{q(k+1)}^{ - 1} + \xi \bm{I}_n)^{ - 1}}<\frac{{{\sigma _{\max }^2}({\bar x _{q(k)}^o})}}{{\xi  + {\sigma _{\max }^2}({\bar x _{q(k)}^o})}}\bm{I}_n\nonumber
	\end{equation}
	are always valid.
\end{lemma}
\begin{proof}
According to Lemma \ref{Lemma 1}, let $Z=L_{q(k)}$, $U=({\bar x _{q(k)}^o})^T$, $W=\bm{I}_{n^2}$, $V={\bar x _{q(k)}^o}$, one has 
\begin{equation}
	L_{q(k+1)}=(L_{q(k)}^{-1}+({\bar x _{q(k)}^o})^T{\bar x _{q(k)}^o})^{-1}.
\end{equation}

Then, we have $L_{q(k+1)}^{-1}=L_{q(k)}^{-1}+({\bar x _{q(k)}^o})^T{\bar x _{q(k)}^o}>{\bar x _{q(k)}^o})^T{\bar x _{q(k)}^o}$. Thus, $\bar U{(L_{q(k+1)}^{ - 1} + \xi \bm{I}_n)^{ - 1}}<\bar U{(\bar U + \xi \bm{I}_n)^{ - 1}}$ where $\bar U=({\bar x _{q(k)}^o})^T{\bar x _{q(k)}^o}$. Call Lemma \ref{Lemma 1} again, we have $\bar U{(\bar U + \xi \bm{I}_n)^{ - 1}}=\bar U({\bar U^{ - 1}} - {\bar U^{ - 1}}{(\bar U{(\xi \bm{I}_n)^{ - 1}} + \bm{I}_n)^{ - 1}})=\bm{I}_n - {(\bar U{(\xi \bm{I}_n)^{ - 1}} + \bm{I}_n)^{ - 1}}<\bm{I}_n - {({\lambda _{\max }}(\bar U){(\xi \bm{I}_n)^{ - 1}} + \bm{I}_n)^{ - 1}}=\frac{{{\lambda _{\max }}(\bar U)}}{{\xi  + {\lambda _{\max }}(\bar U)}}\bm{I}_n=\frac{{{\sigma _{\max }^2}({\bar x _{q(k)}^o})}}{{\xi  + {\sigma _{\max }^2}({\bar x _{q(k)}^o})}}\bm{I}_n\nonumber$. The proof is completed. 
\end{proof} 

Denote ${\hat x _{l(k)}^o}=col({\hat x _{N+1(k)}^o},...,{\hat x _{N+M(k)}^o})$ and ${\mathcal{H}}=diag(g_{N+1}^0,...,g_{N+M}^0)$. Then according to the designed observer, the global estimation error ${\tilde x _{l(k)}^o}={\hat x _{l(k)}^o}-\bm{1}_M\otimes x_{(k)}^o$ of leaders can be transformed into
\begin{equation}
	{\tilde x _{l(k+1)}^o}= S_l{\tilde x _{l(k)}^o}+ {\zeta _{l(k)}^T}\tilde A_{l(k)}^{vec},
\end{equation}
where  $\tilde A_{l(k)}^{vec}$ is the vector representation of $diag(\tilde A_{N+1,0(k)},...,\tilde A_{N+M,0(k)})$ with $\tilde A_{q,0(k)}=\hat A_{q,0(k)}-A_0$, $S_l ={\bm{I}_M} \otimes A_0 - {\mu}_l({{\mathcal L}_3}+{\mathcal{H}}) \otimes F_l$ and  ${\zeta _{l(k)}}=diag({\bar x _{N+1(k)}^o},...,{\bar x _{N+M(k)}^o})$.  According to Lemma 3 of \cite{KIUMARSI201786}, we can choose suitable $\mu_l$ and $F_l$ to ensure that $S_l$ is Schur matrix. The following theorem provides a range of values for the parameter $\Lambda_l$ with respect to the observation error achieving asymptotic convergence. 
\begin{theorem} \label{Theorem 2}
	Consider the heterogeneous MAS \eqref{1}-\eqref{4} with Assumptions \ref{Assumption 1} holds. Then, the estimation error $\mathop {\lim }\limits_{k \to \infty } \tilde x_{q(k)}^o = {\hat x_{q(k)}^o-x_{(k)}^o}=\bm{0}$ for $\forall q \in \mathbb{L} $ if  the coupling gain ${\Lambda}_l$ satisfies 
	\begin{equation}
		0 < \Lambda_l  \le \frac{{{\lambda _{\min }}(W_l - S_l^TW_l{S_l})}}{{\sigma _{\max }^2({{\mathcal L}_3}+{\mathcal{H}}){\Gamma _l}}},
	\end{equation}
	where $W_l=\dfrac{1}{2}{{\bar L}_l}(({{\mathcal L}_3}+{\mathcal{H}})\otimes {\bm{I}_n}) + \dfrac{1}{2}(({{\mathcal L}_3}+{\mathcal{H}})\otimes {\bm{I}_n})^T\bar L_l$ with ${{\bar L}_l}=diag({(L_{q(k+1)}^{ - 1} + \xi \bm{I}_n)^{ - 1}})$ and $\Gamma_l=\frac{{{\sigma _{\max }^2}({\zeta _{l(k)}})}}{{\xi  + {\sigma _{\max }^2}({\zeta _{l(k)}})}}$.
\end{theorem}
\begin{proof}
Consider the Lyapunov function 
\begin{equation}
	{V_{(k)}} = ({\tilde x _{l(k)}^o})^TW_l{\tilde x _{l(k)}^o} + (\tilde A_{l(k)}^{vec})^T{\Lambda_l ^{ - 1}}\tilde A_{l(k)}^{vec}.
\end{equation}

Using \eqref{25}-\eqref{27} and Lemma \ref{Lemma 5}, we have 
\begin{align}
	&\Delta {V_{(k)}}= - ({\tilde x _{l(k)}^o})^T(W_l-S_l^TW_lS_l){\tilde x _{l(k)}^o}\nonumber \\
	&~~~~~+\Lambda_l({\tilde x _{l(k)}^o})^TS_l^T(({{\mathcal L}_3}+{\mathcal{H}})\otimes {\bm{I}_n})^T\bar L_l^T{\zeta _{l(k)}^T}\nonumber \\
	&~~~~~~~~\times{\zeta _{l(k)}}{{\bar L}_l}(({{\mathcal L}_3}+{\mathcal{H}})\otimes {\bm{I}_n})S_l{\tilde x _{l(k)}^o}\nonumber \\
	&~+2(\tilde A_{l(k)}^{vec})^T{\zeta _{l(k)}}W_lS_l{\tilde x _{l(k)}^o}+(\tilde A_{l(k)}^{vec})^T{\zeta _{l(k)}}W_l{\zeta _{l(k)}^T}\tilde A_{l(k)}^{vec}\nonumber \\
	&~~~~~-(\tilde A_{l(k)}^{vec})^T{\zeta _{l(k)}}{{\bar L}_l}(({{\mathcal L}_3}+{\mathcal{H}})\otimes {\bm{I}_n}){\zeta _{l(k)}^T}\tilde A_{l(k)}^{vec}\nonumber \\
	&~~~~~-(\tilde A_{l(k)}^{vec})^T{\zeta _{l(k)}}(({{\mathcal L}_3}+{\mathcal{H}})\otimes {\bm{I}_n})^T{{\bar L}_l}{\zeta _{l(k)}^T}\tilde A_{l(k)}^{vec}\nonumber \\
	&~~~~~-(\tilde A_{l(k)}^{vec})^T{\zeta _{l(k)}}({{\bar L}_l}(({{\mathcal L}_3}+{\mathcal{H}})\otimes {\bm{I}_n})\nonumber \\
	&~~~~~~~~ + (({{\mathcal L}_3}+{\mathcal{H}})\otimes {\bm{I}_n})^T\bar L_l^T)S_l{\tilde x _{l(k)}^o}\nonumber \\
	&~~~~~+\Lambda_l(\tilde A_{l(k)}^{vec})^T {\zeta _{l(k)}}(({{\mathcal L}_3}+{\mathcal{H}})\otimes {\bm{I}_n})^T\bar L_l^T{\zeta _{l(k)}^T}\nonumber \\
	&~~~~~~~~\times{\zeta _{l(k)}}{{\bar L}_l}(({{\mathcal L}_3}+{\mathcal{H}})\otimes {\bm{I}_n}){\zeta _{l(k)}^T}\tilde A_{l(k)}^{vec}\nonumber \\
	&~~~~~+2\Lambda_l (\tilde A_{l(k)}^{vec})^T{\zeta _{l(k)}}(({{\mathcal L}_3}+{\mathcal{H}})\otimes {\bm{I}_n})^T\bar L_l^T{\zeta _{l(k)}^T}\nonumber \\
	&~~~~~~~~\times{\zeta _{l(k)}}{{\bar L}_l}(({{\mathcal L}_3}+{\mathcal{H}})\otimes {\bm{I}_n})S_l{\tilde x _{l(k)}^o}.
\end{align}

Let  $W_l=\dfrac{1}{2}{{\bar L}_l}(({{\mathcal L}_3}+{\mathcal{H}})\otimes {\bm{I}_n}) + \dfrac{1}{2}(({{\mathcal L}_3}+{\mathcal{H}})\otimes {\bm{I}_n})^T\bar L_l$, then we have
\begin{align}
	&\Delta {V_{(k)}}= - ({\tilde x _{l(k)}^o})^T(W_l-S_l^TW_lS_l){\tilde x _{l(k)}^o}\nonumber \\
	&~~~~~+\Lambda_l({\tilde x _{l(k)}^o})^TS_l^T(({{\mathcal L}_3}+{\mathcal{H}})\otimes {\bm{I}_n})^T\bar L_l^T{\zeta _{l(k)}^T}\nonumber \\
	&~~~~~~~~\times{\zeta _{l(k)}}{{\bar L}_l}(({{\mathcal L}_3}+{\mathcal{H}})\otimes {\bm{I}_n})S_l{\tilde x _{l(k)}^o}\nonumber \\
	&~~~~~-(\tilde A_{l(k)}^{vec})^T{\zeta _{l(k)}}W_l{\zeta _{l(k)}^T}\tilde A_{l(k)}^{vec}\nonumber \\
	&~~~~~+\Lambda_l(\tilde A_{l(k)}^{vec})^T {\zeta _{l(k)}}(({{\mathcal L}_3}+{\mathcal{H}})\otimes {\bm{I}_n})^T\bar L_l^T{\zeta _{l(k)}^T}\nonumber \\
	&~~~~~~~~\times{\zeta _{l(k)}}{{\bar L}_l}(({{\mathcal L}_3}+{\mathcal{H}})\otimes {\bm{I}_n}){\zeta _{l(k)}^T}\tilde A_{l(k)}^{vec}\nonumber \\
	&~~~~~+2\Lambda_l (\tilde A_{l(k)}^{vec})^T{\zeta _{l(k)}}(({{\mathcal L}_3}+{\mathcal{H}})\otimes {\bm{I}_n})^T\bar L_l^T{\zeta _{l(k)}^T}\nonumber \\
	&~~~~~~~~\times{\zeta _{l(k)}}{{\bar L}_l}(({{\mathcal L}_3}+{\mathcal{H}})\otimes {\bm{I}_n})S_l{\tilde x _{l(k)}^o}.
\end{align}

By using Lemma \ref{Lemma 5}, Young’s inequality and the fact of ${{\bar L}_l}<\bm{I}_{nM}$, one has
\begin{align}
	&\Delta {V_{(k)}}< - ({\tilde x _{l(k)}^o})^T(W_l-S_l^TW_lS_l\nonumber \\
	&-2\Lambda_l \Gamma_l (({{\mathcal L}_3}+{\mathcal{H}})\otimes {\bm{I}_n})^T(({{\mathcal L}_3}+{\mathcal{H}})\otimes {\bm{I}_n})){\tilde x _{l(k)}^o}\nonumber \\
	&-(\tilde A_{l(k)}^{vec})^T{\zeta _{l(k)}}(W_l-2\Lambda_l\Gamma_l(({{\mathcal L}_3}+{\mathcal{H}})\otimes {\bm{I}_n})^T\nonumber \\
	&~~~~\times(({{\mathcal L}_3}+{\mathcal{H}})\otimes {\bm{I}_n})){\zeta _{l(k)}^T}\tilde A_{l(k)}^{vec},
\end{align}
where $\Gamma_l=\frac{{{\sigma _{\max }^2}({\zeta _{l(k)}})}}{{\xi  + {\sigma _{\max }^2}({\zeta _{l(k)}})}}$. Thus, $\Delta {V_{(k)}}<0$ if the following inequalities hold:
\begin{numcases}{}
	W_l-S_l^TW_lS_l \nonumber\\
	~~~-2\Lambda_l \Gamma_l (({{\mathcal L}_3}+{\mathcal{H}})\otimes {\bm{I}_n})^T(({{\mathcal L}_3}+{\mathcal{H}})\otimes {\bm{I}_n})\geq0,\nonumber\\[2mm]
	W_l-2\Lambda_l\Gamma_l(({{\mathcal L}_3}+{\mathcal{H}})\otimes {\bm{I}_n})^T(({{\mathcal L}_3}+{\mathcal{H}})\otimes {\bm{I}_n})\geq0\nonumber	.
\end{numcases}
\indent
Therefore, the above conditions are  satisfied if 
\begin{equation}
	0 < \Lambda  \le \frac{{{\lambda _{\min }}(W_l - S_l^TW_l{S_l})}}{{\sigma _{\max }^2({{\mathcal L}_3}+{\mathcal{H}}){\Gamma _l}}},
\end{equation}
and the estimation error $\mathop {\lim }\limits_{k \to \infty } \tilde x_{q(k)}^o = {\hat x_{q(k)}^o-x_{(k)}^o}=\bm{0}$ for $\forall q \in \mathbb{L} $. The proof is completed.
\end{proof}

\textbf{Observer for tracking leader by followers.} For follower $i$, the following observer is designed to estimate $x_{(k)}^o$:
\begin{numcases}{}
	{L_{i(k + 1)}^o} = {L_{i(k)}^o} - {L_{i(k)}^o}({\bar x _{i(k)}^o})^T\nonumber\\
	~~~~~~~~~~~~\times{(\bm{I}_{n^2} + {\bar x _{i(k)}^o}{L_{i(k)}^o} ({\bar x _{i(k)}^o})^T)^{ - 1}} {\bar x _{i(k)}^o}{L_{i(k)}^o},\label{30}\\[2mm]
	\hat A_{i,0(k + 1)}^{vec} \nonumber\\
	~~~~~~= \hat A_{i,0(k)}^{vec}- {\Lambda}_f^o{\bar x _{i(k)}^o}{(L_{i(k+1)}^{ - 1} + \xi \bm{I}_n)^{ - 1}}{\eta _{i(k + 1)}^o},  \\[2mm]
	{\hat x _{i(k + 1)}^o} = \hat A_{i,0(k)}{\hat x _{i(k)}^o} - \mu_f^oF_f^o {\eta _{i(k)}^o},\label{32}
\end{numcases}
where  ${\bar x _{i(k)}^o}={\bm{I}_n} \otimes {\hat x _{i(k)}^o}$, ${\eta _{i(k)}^o} = \sum\limits_{j=1}^{N} {{a_{ij}}({\hat x _{i(k)}^o} - {\hat x _{j(k)}^o})}  + \sum\limits_{j = N+1}^{N+M} {g_i^j ({\hat x _{i(k)}^o} - {\hat x _{j(k)}^o})} $, $\xi\geq 1$, ${\Lambda}_f^o>0$ and $\mu_f^o > 0$ are constant values, $F_f^o$ is constant matrix.
\\ \indent
Denote ${\hat x _{f(k)}^o}=col({\hat x _{1(k)}^o},...,{\hat x _{N(k)}^o})$. Since $\mathop {\lim }\limits_{k \to \infty }{\hat x _{l(k)}^o}=\bm{1}_M\otimes x_{(k)}^o$ in Theorem \ref{Theorem 2} and the row sum of ${\mathcal{L}_1^{-1}}{\mathcal{L}_2}\otimes \bm{I}_n$ is $-1$ . Then the global estimation error ${\tilde x _{f(k)}^o}={\hat x _{f(k)}^o}-\bm{1}_M\otimes x_{(k)}^o={\hat x _{f(k)}^o}+ ({\mathcal{L}_1^{-1}}{\mathcal{L}_2}\otimes \bm{I}_n)\hat x_{l(k)}^o$ can be transformed into
\begin{equation}
	{\tilde x _{f(k+1)}^o}= S_f^o{\tilde x _{f(k)}^o}+ {\zeta _{f(k)}^{o^T}}\tilde A_{f0(k)}^{vec},
\end{equation}
where  $\tilde A_{f0(k)}^{vec}$ is the vector representation of $diag(\tilde A_{1,0(k)},...,\tilde A_{N,0(k)})$ with $\tilde A_{i,0(k)}=\hat A_{i,0(k)}-A_0$, $S_f^o ={\bm{I}_N} \otimes A_0 - {\mu}_f^o{{\mathcal L}_1} \otimes F_f^o$ and  ${\zeta _{f(k)}^o}=diag({\bar x _{1(k)}^o},...,{\bar x _{N(k)}^o})$.

\textbf{Observer for formation state by agents.} Due to the presence of ITFLs, not only followers need to observe $h_{q(k)}^l$, but ITFLs also need to observe to relay information. Thus for agent $m$, the following observer is designed to estimate $h_{q(k)}^l,q\in\mathcal{N}_{m(k)}^F$ or $q\in\mathcal{N}_{m(k)}^L$:
\begin{numcases}{}
	{L_{m(k + 1)}^q} = {L_{m(k)}^q} - {L_{m(k)}^q}({\bar h _{m(k)}^q})^T\nonumber\\
	~~~~~~~~\times{(\bm{I}_{n^2} + {\bar h _{m(k)}^q}{L_{m(k)}^q} ({\bar h _{m(k)}^q})^T)^{ - 1}} {\bar h _{m(k)}^q}{L_{m(k)}^q},\label{34}\\[2mm]
	\hat A_{m,q(k + 1)}^{vec} \nonumber\\
	~~~~= \hat A_{m,q(k)}^{vec}- {\Lambda}_f^q{\bar h _{m(k)}^q}{(L_{m(k+1)}^{q^ {-1}} + \xi \bm{I}_n)^{ - 1}}{\eta _{m(k + 1)}^q},  \\[2mm]
	{\hat h _{m(k + 1)}^q} = \hat A_{m,q(k)}{\hat h _{m(k)}^q} - \mu_f^qF_f^q {\eta _{m(k)}^q},\label{36}
\end{numcases}
where  ${\bar h _{m(k)}^q}={\bm{I}_n} \otimes {\hat h _{m(k)}^q}$, $\eta _{m(k)}^q = \sum\limits_{q \ne j = N+1}^{N+M} {{a_{mj}}{{\bar a}_{jq(k)}}(\hat h_{m(k)}^q - \hat h_{j(k)}^q)}  + g_m^q(\hat h_{m(k)}^q - h_{q(k)}^l)$ if agent $m$ is formation leader and $\eta _{m(k)}^q = \sum\limits_{ j = N+1}^{N+M} {g_m^j{{\bar a}_{jq(k)}}(\hat h_{m(k)}^q - \hat h_{j(k)}^q)}  + \sum\limits_{q \ne j = 1}^{ N} {{a_{mj}}\bar g_{j(k)}^q(\hat h_{m(k)}^q - \hat h_{j(k)}^q)}  + g_m^q(\hat h_{m(k)}^q - h_{q(k)}^l)$ if agent $m$ is follower, where $\bar g_{j(k)}^q$ and ${{\bar a}_{jq(k)}}$ are obtained by Algorithm \ref{Algorithm 1}, $\xi\geq 1$, ${\Lambda}_f^q>0$ and $\mu_f^q > 0$ are constant values. 

In the definition of $\eta _{m(k)}^q$, the agent $m$ checks whether its neighbors are influenced by the formation leader $q$ (by multiplying incoming data by $\bar g_{j(k)}^q$ or ${{\bar a}_{jq(k)}}$) before receiving information. In other words, only the directed graph with leader $q$ as the root node and the agents influenced by it as the child nodes are retained. After sufficient steps, denote the vector combination of the observations and the Laplacian matrix of these child nodes as ${\hat h _{f(k)}^q}=col({\hat h _{m(k)}^q})$ and $\mathcal{L}^q\in \mathbb{R}^{v_q\times v_q}$, respectively, where $v_q$ represents the number of agents influenced by formation leader $q$. Let ${\mathcal{H}}^q=diag(a_{mq})$. Then the global estimation error ${\tilde h _{f(k)}^q}={\hat h _{f(k)}^q}-\bm{1}_{v_q}\otimes h_{q(k)}^l$ can be transformed into
\begin{equation}
	{\tilde h _{f(k+1)}^q}= S_f^q{\tilde h _{f(k)}^q}+ {\zeta _{f(k)}^{q^T}}\tilde A_{fq(k)}^{vec},
\end{equation}
where  $\tilde A_{fq(k)}^{vec}$ is the vector representation of $diag(\tilde A_{m,q(k)})$ with $\tilde A_{m,q(k)}=\hat A_{m,q(k)}-S_q$, $S_f^q ={\bm{I}_{v_q}} \otimes S_q - {\mu}_f^q{{\mathcal L}^q} \otimes F_f^q$ and  ${\zeta _{f(k)}^q}=diag({\bar h _{m(k)}^q})$. Similarly, according to Lemma 3 of \cite{KIUMARSI201786}, we can choose suitable $\mu_f^o$, $\mu_f^q$, $F_f^o$ and $F_f^q$ to ensure that $S_f^o$ and $S_f^q$ are Schur matrices.
\begin{theorem} \label{Theorem 3}
	Consider the heterogeneous MAS \eqref{1}-\eqref{4} with Assumptions \ref{Assumption 1} holds. Then, the estimation error $\mathop {\lim }\limits_{k \to \infty } \tilde x_{i(k)}^o = {\hat x_{i(k)}^o-x_{(k)}^o}=\bm{0}$ for $\forall i \in \mathbb{F}$ and $\mathop {\lim }\limits_{k \to \infty } \tilde h_{i(k)}^{q}={\hat h_{i(k)}^{q}-h_{q(k)}^{l}}= \bm{0}$ for $\forall  q\in {\mathcal N}_{i}^F, i \in \mathbb{F}$ if  the coupling gains ${\Lambda}_f^o$ and ${\Lambda}_f^q$ satisfy 
	\begin{numcases}{}
		0 < \Lambda_f^o  \le \frac{{{\lambda _{\min }}(W_f^o - S_f^{o^T}W_f^o{S_f^o})}}{{\sigma _{\max }^2({{\mathcal L}_1}){\Gamma _f^o}}},\\
		0 < \Lambda_f^q  \le \frac{{{\lambda _{\min }}(W_f^q - S_f^{q^T}W_f^q{S_f^q})}}{{\sigma _{\max }^2({{\mathcal L}^q}+{{\mathcal H}^q}){\Gamma _f^o}}},
	\end{numcases}
	where $W_f^o=\dfrac{1}{2}{{\bar L}_f^o}({{\mathcal L}_1}\otimes {\bm{I}_n}) + \dfrac{1}{2}({{\mathcal L}_1}\otimes {\bm{I}_n})^T\bar L_f^o$ with ${{\bar L}_f^o}=diag({(L_{i(k+1)}^{ - 1} + \xi \bm{I}_n)^{ - 1}})$ and $\Gamma_f^o=\frac{{{\sigma _{\max }^2}({\zeta _{f(k)}^o})}}{{\xi  + {\sigma _{\max }^2}({\zeta _{f(k)}^o})}}$, $W_f^q=\dfrac{1}{2}{{\bar L}_f^q}(({{\mathcal L}^q}+{{\mathcal H}^q})\otimes {\bm{I}_n}) + \dfrac{1}{2}(({{\mathcal L}^q}+{{\mathcal H}^q})\otimes {\bm{I}_n})^T\bar L_f^q$ with ${{\bar L}_f^q}=diag({(L_{m(k+1)}^{q^ {- 1}} + \xi \bm{I}_n)^{ - 1}})$ and $\Gamma_f^q=\frac{{{\sigma _{\max }^2}({\zeta _{f(k)}^q})}}{{\xi  + {\sigma _{\max }^2}({\zeta _{f(k)}^q})}}$.
\end{theorem}
\begin{proof}
The proof is similar to Theorem \ref{Theorem 2} and omitted.
\end{proof} 
\begin{remark}
In contrast to prior works \cite{HAGHSHENAS2015210,8277155,8869852,9627528,9802518,KIUMARSI201786}, which do not account for the presence of ITFLs in the topology graph, this paper introduces the concept of ITFL and proposes an observer that effectively solves the problem of underutilizing leader information. Furthermore, this data-based observer eliminates the need for a single leader or homogeneous leaders, distinguishing it from the observer proposed in \cite{KIUMARSI201786}.	
\end{remark}
\begin{remark}
Please note that the upper bounds of the coupling gains, denoted as ${\Lambda}_l,{\Lambda}_f^o,{\Lambda}_f^q$ in Theorems 2 and 3, are positively correlated with $\xi$. To ensure the observers' convergence, values close to $0$ can be selected for these coupling gains, avoiding the need to calculate the exact upper bound.
\end{remark}
\subsection{Online Data-Driven Learning Method for PFCC Problem}
In this subsection, we present a learning algorithm for solving the PFCC problem using solely measurement data, eliminating the need for knowledge about the system dynamics. We reformulate the Riccati equations \eqref{21} and \eqref{22} as  data-based  linear matrix equations, enabling the estimation of unknown matrices and models using least squares. The controller's learning process is implemented through value iteration, and Theorem \ref{Theorem 4} verifies the effectiveness of the data-driven algorithm.

Based on \eqref{1}, \eqref{2}, \eqref{21} and \eqref{22}, one gets
\begin{numcases}{}
	\hat X_{q(k)}^{l^T}{P_q}\hat X_{q(k)}^l =\hat X_{q(k)}^{l^T} C_l^T{Q_q}{C_l}\hat X_{q(k)}^{l} \nonumber\\
	~~~~~~~~~~~~~~~~~~~~~~~~+\hat X_{q(k+1)}^{l^T}{P_q}\hat X_{q(k+1)}^{l},~~~~~ q\in\mathbb{L} \label{40}\\ [2mm]
	\hat X_{i(k)}^{T}{P_{i(k)}}\hat X_{i(k)} =\hat X_{i(k)}^T C_{i(k)}^T{Q_i}{C_{i(k)}}\hat X_{i(k)} \nonumber\\
	~~~~~~~~~~~~~~~~~~~~~~~~~+  \hat X_{i(k+1)}^T{P_{i(k)}}\hat X_{i(k+1)}, ~~ i\in\mathbb{F}.
\end{numcases}
and 
\begin{numcases}{}
	\hat X_{q(k+1)}^{l^T}{P_q}\hat X_{q(k+1)}^l =\hat X_{q(k)}^{l^T} \bar A_q^T{P_q}{\bar A_q}\hat X_{q(k)}^{l} \nonumber\\
	+2u_{q(k)}^{l^T} \bar B_q^T{P_q}{\bar A_q}\hat X_{q(k)}^{l}+u_{q(k)}^{l^T} \bar B_q^T{P_q}{\bar B_q}u_{q(k)}^{l},~ q\in\mathbb{L} \label{42}\\ [2mm]
	\hat X_{i(k+1)}^{T}{P_{i(k)}}\hat X_{i(k+1)} =\hat X_{i(k)}^T \bar A_{i(k)}^T{P_{i(k)}}{\bar A_{i(k)}}\hat X_{i(k)} \nonumber\\
	~~~~~~~~~~~~~~~~~~~~~~+2u_{i(k)}^{T} \bar B_{i(k)}^T{P_{i(k)}}{\bar A_{i(k)}}\hat X_{i(k)}\nonumber\\
	~~~~~~~~~~~~~~~~~~~~~~+u_{i(k)}^{T} \bar B_{i(k)}^T{P_{i(k)}}{\bar B_{i(k)}}u_{i(k)}, ~~ i\in\mathbb{F}.
\end{numcases}

Define the following operators:
\begin{numcases}{}
	{\psi_{q(k)}^l}=\text{vecv}{(\hat X_{q(k)}^l)}\nonumber\\
	\tau _{q(k)}^l=\hat X_{q(k)}^l \otimes u_{q(k)}^{l}\nonumber\\
	\omega_{q(k)}^l=\text{vecv}{(u_{q(k)}^l)}\nonumber\\
	\Psi _{q}^l = {\left[ {\begin{array}{*{20}{c}}
				{\psi _{q(k)}^l}&{\psi _{q(k - 1)}^l}& \cdots &{\psi _{q(k - s - 1)}^l}
		\end{array}} \right]^T}\nonumber\\
	\Psi _{q+}^l = {\left[ {\begin{array}{*{20}{c}}
				{\psi _{q(k+1)}^l}&{\psi _{q(k)}^l}& \cdots &{\psi _{q(k - s)}^l}
		\end{array}} \right]^T}\nonumber\\
	{\rm T}_q^l = {\left[ {\begin{array}{*{20}{c}}
				{\tau _{q(k)}^l}&{\tau _{q(k - 1)}^l}& \cdots &{\tau _{q(k - s - 1)}^l}
		\end{array}} \right]^T}\nonumber\\
	\Omega _q^l = {\left[ {\begin{array}{*{20}{c}}
				{\omega _{q(k)}^l}&{\omega _{q(k - 1)}^l}& \cdots &{\omega _{q(k - s - 1)}^l}
		\end{array}} \right]^T}\nonumber\\
	\Theta _q^l = \left[ {\begin{array}{*{20}{c}}
			{\Psi _{q}^l}&{2{\rm T}_q^l}&{\Omega _q^l}
	\end{array}} \right]\nonumber,
\end{numcases}
where $s$ represents the dimension of the collected data. Using the same method to define ${\psi_{i(k)}}$, $\tau _{i(k)}$, $\omega_{i(k)}$, $\Psi _{i}$, $\Psi _{i+}$, ${\rm T}_i$, $\Omega _i$ and $\Theta_i$. Then, according to the property of Kronecker product, \eqref{40} and \eqref{42} can be transformed into
\begin{numcases}{}
	{\psi_{q(k)}^{l^T}} \text{vecm}({P_q})=\psi_{q(k)}^{l^T}  \text{vecm}(C_l^T{Q_q}{C_l})\nonumber\\
	~~~~~~~~~~~~~~~~~~~~~~~~+\psi_{q(k+1)}^{l^T}  \text{vecm}(P_q), \label{44}\\ [2mm]
	\psi_{q(k+1)}^{l^T} \text{vecm}({P_q})=\psi_{q(k)}^{l^T}  \text{vecm}(\bar A_q^T{P_q}{\bar A_q})\nonumber\\
	~~~+2\tau _{q(k)}^{l^T}\text{vec}(\bar B_q^T{P_q}{\bar A_q}) +\omega_{q(k)}^{l^T}  \text{vecm}(\bar B_q^T{P_q}{\bar B_q}).\label{45}
\end{numcases}

Using the collected data,  \eqref{44} and \eqref{45} can be sorted into the following linear matrix equations
\begin{numcases}{}
	\Psi _{q}^l{\text{vecm}}({P_q}) = \Phi _q^l, \\ [2mm]
	\Theta _q^l\Xi _q = \Psi _{q+}^l{\text{vecm}}({P_q}),
\end{numcases}
where 
\begin{numcases}{}
	\Phi _q^l = \Psi_{q}^{l}  \text{vecm}(C_l^T{Q_q}{C_l})+\Psi_{q+}^{l}  \text{vecm}(P_q), \nonumber \\ [2mm]
	\Xi _q = {\left[ {\begin{array}{*{20}{c}}
				{\text{vecm}^T}(\Xi _{q,1})&{\text{vec}^T}(\Xi _{q,2})&{\text{vecm}^T}(\Xi _{q,3})
		\end{array}} \right]^T}, \nonumber \\ [2mm]
	\Xi _{q,1} = \bar A_q^T{P_q}{{\bar A}_q}, ~\Xi _{q,2} = \bar B_q^T{P_q}{{\bar A}_q}, ~\Xi _{q,3} = \bar B_q^T{P_q}{{\bar B}_q}.\nonumber
\end{numcases}

Employing the same methodology to define $\Phi _i$, $\Xi _i$, $\Xi _{i,1}$, $\Xi _{i,2}$ and $\Xi _{i,3}$. We introduce Algorithm \ref{Algorithm 2} to derive a data-driven control solution for the PFCC problem.
\begin{algorithm}[htbp]	
	\renewcommand{\algorithmicrequire}{\textbf{Input:}}
	\caption{Online Data-Driven Learning Algorithm for PFCC Problem}\label{Algorithm 2}
	\begin{algorithmic}[1]
		\Require{$\hat K_q^0$, $\hat K_i^0$, $\hat P_q^0>0$, $\hat P_i^0>0$}
		\State Let $j = 0$.
		\State Update matrices $\hat P_q$ and $\hat P_i$:
		\begin{equation*}
			\left\{
			\begin{array}{ll}
				{\text{vecm}}(\hat P_q^{j + 1}) = {(\Psi _{q}^{{l^T}}\Psi _{q}^l)^{ - 1}}\Psi _{q}^{{l^T}}\Phi _q^l,\\[2mm]
				{\text{vecm}}(\hat P_i^{j + 1}) = {(\Psi _{i}^{{T}}\Psi _{i})^{ - 1}}\Psi _{i}^{{T}}\Phi _i.
			\end{array}\right.
		\end{equation*}
		\State Update model vectors $\Xi _q$ and $\Xi _i$:
		\begin{equation*}
			\left\{
			\begin{array}{ll}
				\Xi _q^j = {(\Theta _q^{{l^T}}\Theta _q^l)^{ - 1}}\Theta _q^{{l^T}}\Psi _{q + }^l{\text{vecm}}(\hat P_q^{j + 1}),\\[2mm]
				\Xi _i^j = {(\Theta _i^{{T}}\Theta _i)^{ - 1}}\Theta _i^{{T}}\Psi _{i + }{\text{vecm}}(\hat P_i^{j + 1}).
			\end{array}\right.
		\end{equation*}
		\State  Update control gains $\hat K_q$ and $\hat K_i$:
		\begin{equation*} 
			\left\{
			\begin{array}{ll}
				\hat K_q^{j + 1} ={(\Xi _{q,3}^j)^ + }\Xi _{q,2}^j,\\[2mm]
				\hat K_i^{j + 1} = {(\Xi _{i,3}^j)^ + }\Xi _{i,2}^j.
			\end{array}\right.
		\end{equation*}
		\State Let $j=j+1$ and return to step 2.
	\end{algorithmic}
\end{algorithm}
\begin{theorem} \label{Theorem 4}
	Consider the heterogeneous MAS \eqref{1}-\eqref{4} with Assumptions \ref{Assumption 1}-\ref{Assumption 6} hold. By employing Algorithm \ref{Algorithm 2}, we have
	\begin{equation}
		\mathop {\lim }\limits_{j \to \infty } \hat K_i^j = K_{i(\infty )}^F, \mathop {\lim }\limits_{j \to \infty } \hat K_q^j = K_q^L
	\end{equation}
	and
	\begin{equation}
		\mathop {\lim }\limits_{j \to \infty } \hat P_i^j = P_{i(\infty )}, \mathop {\lim }\limits_{j \to \infty } \hat P_q^j = P_q.
	\end{equation}
\end{theorem}
\begin{proof}
	After $N+M-1$ steps and the convergence of the observers, steps $2-4$ of Algorithm \ref{Algorithm 2} become equivalent to the following iterative steps for agent $m$, $\forall m \in \mathbb{F}\cup\mathbb{L}$: 
	\begin{numcases}{}
		P_m^{j + 1} = C_m^T{Q_m}{C_m} + {({{\bar A}_m} + {{\bar B}_m}K_m^j)^T}P_m^j({{\bar A}_m} + {{\bar B}_m}K_m^j),\nonumber\\[2mm]
		K_m^{j + 1} =  - {(\bar B_m^TP_m^{j + 1}{{\bar B}_m})^ + }\bar B_m^TP_m^{j + 1}{{\bar A}_m},\nonumber
	\end{numcases}
	where ${C_m} = {C_l}$ for $ m \in \mathbb{L}$ and ${C_m} = [{\bm{I}_n}, - \alpha _{i}^{{\phi _{i(\infty )}}(1)}{\bm{I}_n},..., - \alpha _{i}^{{\phi _{i(\infty )}}({{\mathcal I}_{i(\infty )}})}{\bm{I}_n},-{\bm{I}_n}]$ for $ m \in \mathbb{F}$, ${\bar A}_m$ and ${\bar B}_m$ are given by Definition 8. Noted that after $N+M-1$ steps,  ${\bar A}_{i(k)}$ and ${\bar B}_{i(k)}$ of followers remain unchanged, so the subscript $(k)$ is omitted from the above equation. According to Proposition 4.4.1 of \cite{bertsekas2012dynamic}, if $\hat P_q^0$ and $\hat P_i^0$ are positively definite, the above iterative process will converge to the expected value. Hence, the proof is completed. 
\end{proof} 
\begin{remark}
	The data acquisition phase in Algorithm \ref{Algorithm 2} involves introducing noise to meet the persistent excitation (PE) requirement. In this paper, we incorporate normally distributed noise into the control input to ensure data richness during the learning process. The noise input is canceled, and the control process proceeds with the last updated control gain when the norm of the difference between two adjacent control gains falls below a specified threshold. In addition, during the data collection process, the systems employ the previously updated control strategy (with  perturbations added). From these two perspectives, the behavioral policy and the target policy are distinct, thus our approach is off-policy.
\end{remark}

\section{Simulation results}\label{S4}
In this section, we provide an example to validate the accuracy and effectiveness of the proposed data-driven approach for the PFCC problem. We examine a heterogeneous MAS comprising four followers designated as F1, F2, F3, F4, and six formation leaders labeled as L1, L2, L3, L4, L5, and L6. The network topology is visually represented in Fig. \ref{fig2}. In Fig. \ref{fig2}, the sets of IFLs for F1, F2, F3 and F4 are \{L1, L2\}, \{L1, L2, L3, L4\}, \{L1, L3\} and \{L1, L2, L3, L4, L5, L6\}, respectively. Since the information from L1 and L2 cannot be transmitted to F3 and F4 through the followers alone, \{L3, L5\} and \{L4, L6\} serve as the ITFLs for L1 and L2, respectively. Consequently, if observers are not designed for L1 and L2, F3 will be unable to utilize the information from L1, and F4 will be unable to utilize the information from both L1 and L2. The system matrices for the leaders and formation shapes are provided as follows:

\begin{align}
	&{A_{L1}} = \left[ {\begin{array}{*{20}{c}}
			0&1\\
			-2&{4}
	\end{array}} \right],\ \ {B_{L1}} = \left[ {\begin{array}{*{20}{c}}
			{ 0}\\
			{ 2}
	\end{array}} \right],\ \ {S_{L1}} = \left[ {\begin{array}{*{20}{c}}
			0&1\\
			1&{0}
	\end{array}} \right]\nonumber \\
	&{A_{L2}} = \left[ {\begin{array}{*{20}{c}}
			0&1\\
			-1&{1}
	\end{array}} \right],\ \ {B_{L2}} = \left[ {\begin{array}{*{20}{c}}
			0&2\\
			2&{-1}
	\end{array}} \right],\ \ {S_{L2}} = \left[ {\begin{array}{*{20}{c}}
			0&1\\
			1&{0}
	\end{array}} \right]\nonumber \\
	&{A_{L3}} = \left[ {\begin{array}{*{20}{c}}
			0&1\\
			1&{0}
	\end{array}} \right],\ \ {B_{L3}} = \left[ {\begin{array}{*{20}{c}}
			0&2&1\\
			1&{1}&1
	\end{array}} \right],\ \ {S_{L3}} = \left[ {\begin{array}{*{20}{c}}
			0&1\\
			1&{0}
	\end{array}} \right]\nonumber \\
	&{A_{L4}} = \left[ {\begin{array}{*{20}{c}}
			0&1\\
			3&{-2}
	\end{array}} \right],\ \ {B_{L4}} = \left[ {\begin{array}{*{1}{c}}
			{0}\\
			{3}
	\end{array}} \right],\ \ {S_{L4}} = \left[ {\begin{array}{*{20}{c}}
			0&1\\
			1&{0}
	\end{array}} \right]\nonumber \\
	&{A_{L5}} = \left[ {\begin{array}{*{20}{c}}
			0&1\\
			-1&{0}
	\end{array}} \right],\ \ {B_{L5}} = \left[ {\begin{array}{*{20}{c}}
			1&3\\
			2&{0}
	\end{array}} \right],\ \ {S_{L5}} = \left[ {\begin{array}{*{20}{c}}
			0&1\\
			1&{0}
	\end{array}} \right]\nonumber \\
	&{A_{L6}} = \left[ {\begin{array}{*{20}{c}}
			0&1\\
			1&{1}
	\end{array}} \right],\ \ {B_{L6}} = \left[ {\begin{array}{*{20}{c}}
			{ 0}\\
			{ -2}
	\end{array}} \right],\ \ {S_{L6}} = \left[ {\begin{array}{*{20}{c}}
			0&1\\
			1&{0}
	\end{array}} \right].&\nonumber 
\end{align}

The dynamics of followers are given as
\begin{align}
	&{A_{F1}} = \left[ {\begin{array}{*{20}{c}}
			0&1\\
			1&{3}
	\end{array}} \right],\ \ {B_{F1}} = \left[ {\begin{array}{*{20}{c}}
			{ 0}\\
			{ 1}
	\end{array}} \right]\nonumber \\
	&{A_{F2}} = \left[ {\begin{array}{*{20}{c}}
			0&1\\
			2&{1}
	\end{array}} \right],\ \ {B_{F2}} = \left[ {\begin{array}{*{20}{c}}
			{ 0}\\
			{ 2}
	\end{array}} \right]\nonumber \\
	&{A_{F3}} = \left[ {\begin{array}{*{20}{c}}
			0&1\\
			1&{-2}
	\end{array}} \right],\ \ {B_{F3}} = \left[ {\begin{array}{*{20}{c}}
			{ 0}\\
			{-1}
	\end{array}} \right]\nonumber \\
	&{A_{F4}} = \left[ {\begin{array}{*{20}{c}}
			0&1\\
			-1&{2}
	\end{array}} \right],\ \ {B_{F4}} = \left[ {\begin{array}{*{20}{c}}
			{ 0}\\
			{ - 2}
	\end{array}} \right].&\nonumber 
\end{align}

The dynamics of the tracking leader are given as
\begin{align}
	&{A_{0}} = \left[ {\begin{array}{*{20}{c}}
			0&1\\
			1&{0}
	\end{array}} \right].&\nonumber 
\end{align}
\begin{figure}
	\centering 
	\includegraphics[scale=1]{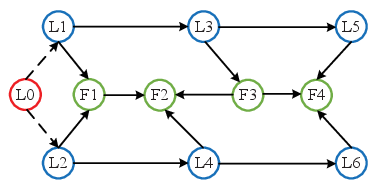}
	\caption {{ Network topology composed of all agents.}}
	\label{fig2}
\end{figure}

In this example, the leaders maintain an identical propensity throughout the first 4000 iterations, with the propensity factors set as ${{\vartheta _1}}=0.1$, ${{\vartheta _2}}=0.1$, ${{\vartheta _3}}=0.1$, ${{\vartheta _4}}=0.1$, ${{\vartheta _5}}=0.1$, ${{\vartheta _6}}=0.1$. After 4000 iterations, the followers are anticipated to prefer L1 more significantly, leading to an adjustment in the propensity factors to ${{\vartheta _1}}=0.5$, ${{\vartheta _2}}=0.1$, ${{\vartheta _3}}=0.1$, ${{\vartheta _4}}=0.1$, ${{\vartheta _5}}=0.1$, ${{\vartheta _6}}=0.1$. Set $Q_{L1}=Q_{L2}=Q_{L3}=Q_{L4}=Q_{L5}=Q_{L6}=2\times I_2$, $Q_{F1}=Q_{F2}=Q_{F3}=Q_{F4}=I_2$. Select $\xi= 4$, $\Lambda_l=8$ in Theorem \ref{Theorem 2} and $\Lambda_f^o =8, \Lambda_f^q =8, \forall  q\in {\mathcal N}_{i}^F, i \in \mathbb{F}$ in Theorem \ref{Theorem 3}. The gains in \eqref{27}, \eqref{32} and \eqref{36} are chosen as $\mu_l =\mu_f^o=\mu_f^{L1}= 0.7$,  $\mu_f^{L2}=\mu_f^{L3}=\mu_f^{L4}=\mu_f^{L6}=\mu_f^{L6}=1$. The gain matrices in \eqref{27}, \eqref{32} and \eqref{36} are chosen as $F_l = F_f^o=F_f^{L3}=\left[ {\begin{array}{*{20}{c}}
		{0.1}&{0.5}\\
		{0.5}&{0.1}
\end{array}} \right], F_f^{L1}=F_f^{L2}=F_f^{L4}=F_f^{L5}=F_f^{L6}=\left[ {\begin{array}{*{20}{c}}
		{0.2}&{0.5}\\
		{0.5}&{0.2}
\end{array}} \right],\forall  q\in {\mathcal N}_{i}^F, i \in \mathbb{F}$. 

\begin{figure}[htbp]
	\centering
	\includegraphics[scale=0.4]{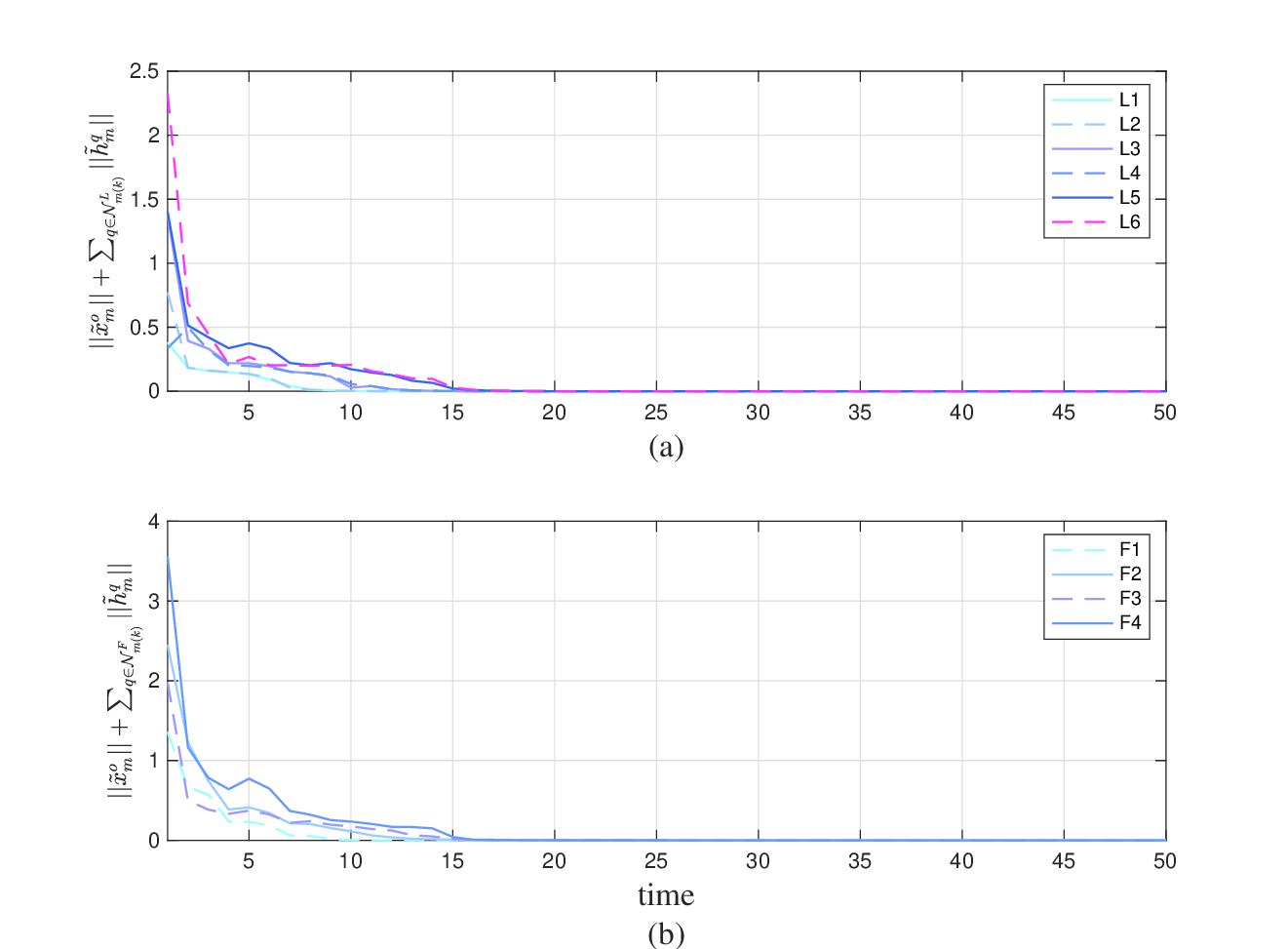}
	\caption {{Sum of the norm of all observer errors for agents. (a). Observer error $||\tilde x_m^o||+\sum\nolimits_{q \in {\mathcal N}_{m}^L} {||\tilde h_m^q||}$ for leaders. (b). Observer error $||\tilde x_m^o||+\sum\nolimits_{q \in {\mathcal N}_{m}^F} {||\tilde h_m^q||}$ for followers.}}
	\label{fig3}
\end{figure}

\begin{figure}[htbp]
	\centering
	\includegraphics[scale=0.4]{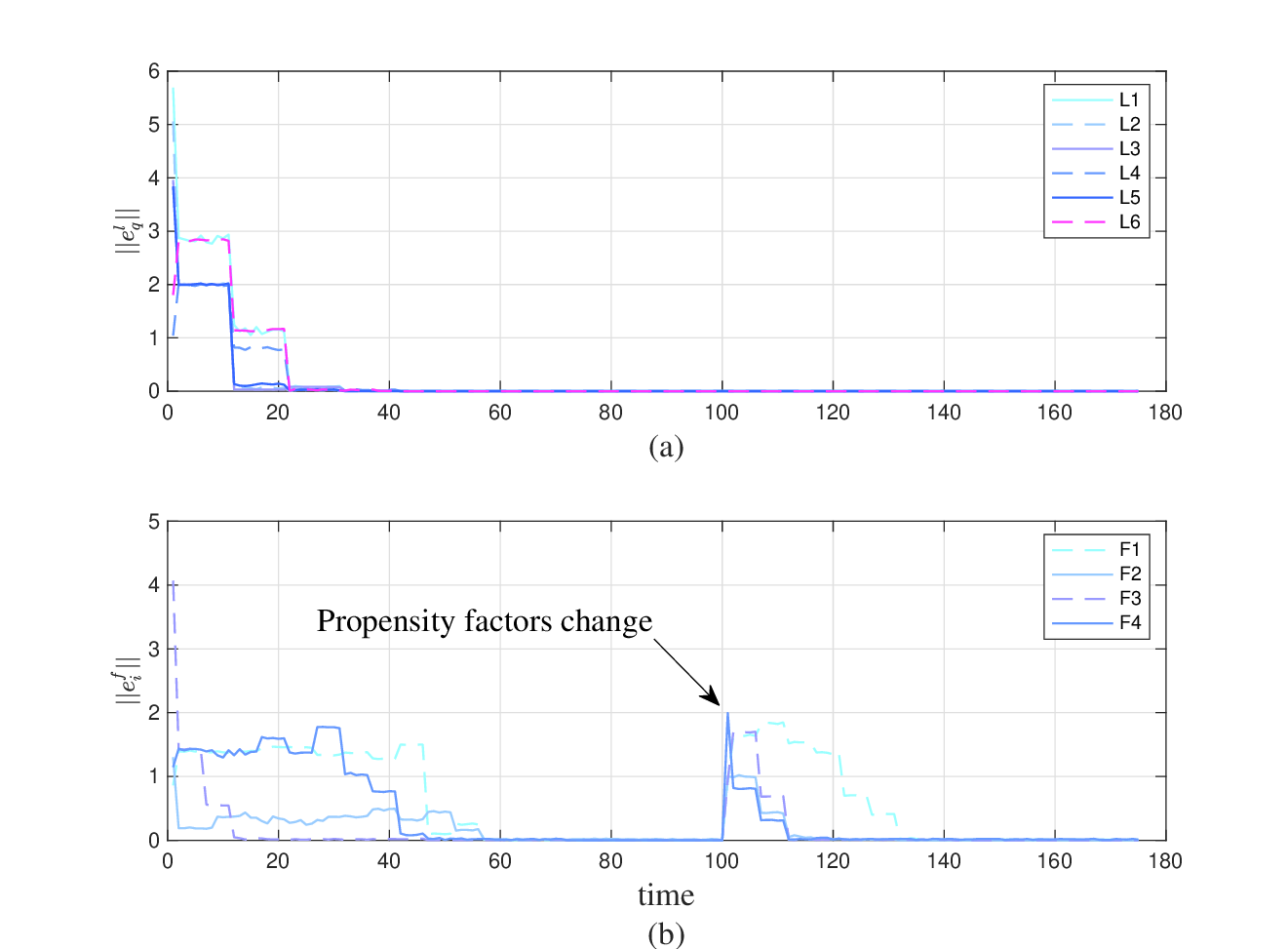}
	\caption {{Norm of errors for PFCC problem. (a). Formation error of leaders $||e_q^l||$in \eqref{58}. (b). Propensity containment error of followers $||e_i^f||$ in \eqref{60}.}}
	\label{fig4}
\end{figure}

The initial state of the tracking leader is set to $[0, 0]^T$, while the initial formation states of L1–L6 are set to $[2, 0]^T$, $[0, 2]^T$, $[0, -2]^T$, $[-2, 0]^T$, $[2, 2]^T$, and $[-2, -2]^T$, respectively. Consequently, the desired formation of the leaders forms a symmetric hexagon. During the simulation, the data in Fig. \ref{fig3} to Fig. \ref{fig5} is sampled every 40 iterations, while the data in Fig. \ref{fig7} and Fig. \ref{fig8} is sampled every 100 iterations. All system matrices assumed to be unknown. Furthermore, to emphasize the advantages of the PFCC method proposed in this paper, the proposed data-driven approach is applied to the traditional two-layer FCC problem \cite{9870036} for comparison. 

Fig. \ref{fig3} shows the sum of the norms of all observer errors for each agent, with Fig. \ref{fig3}(a) corresponding to the formation leaders and Fig. \ref{fig3}(b) to the followers. Fig. \ref{fig4} presents the error norms for the PFCC problem, where Fig. \ref{fig4}(a) pertains to the leaders and Fig. \ref{fig4}(b) to the followers. After 4000 iterations, the propensity factors adjust, enabling the followers to relearn the target controller. These figures demonstrate the successful resolution of the fully heterogeneous PFCC problem.                           

Fig. \ref{fig5} shows the observations made by agents on the formation matrices of leaders L1. In this figure, $\hat A_{m,q}^{ij}$ represents the observed values of agent $m$ for the $i$-th row and $j$-th column of the formation matrix of leader $q$. Our design considerations encompass not only the followers but also the transit formation leaders L3, L5, and L6. This approach ensures that F3 and F4 can effectively utilize the information provided by L1, and similar operations apply to L2, L3 and L4.

Fig. \ref{fig7} shows the 3-D state trajectories of the leaders and followers using the proposed method. Due to the involvement of L3 in the observation of L1, F3 can effectively utilize the information from L1. As a result, the convergence position of F3 is located at the midpoint between L1 and L3, rather than simply tracking L3. In a similar manner, F4 utilizes information from L1, L2, and L4, leading to its final convergence position at the center of all leaders. Fig. \ref{fig8} illustrates the 3-D state trajectory of the traditional FCC approach. In this method, the followers' positions are influenced by the Laplacian matrix, and the fixed topology prevents them from adapting to the leaders' tendencies. However, the PFCC method enables the adjustment of followers' positions after modifying the propensity factors.
\begin{figure}[htbp]
	\centering
	\includegraphics[scale=0.4]{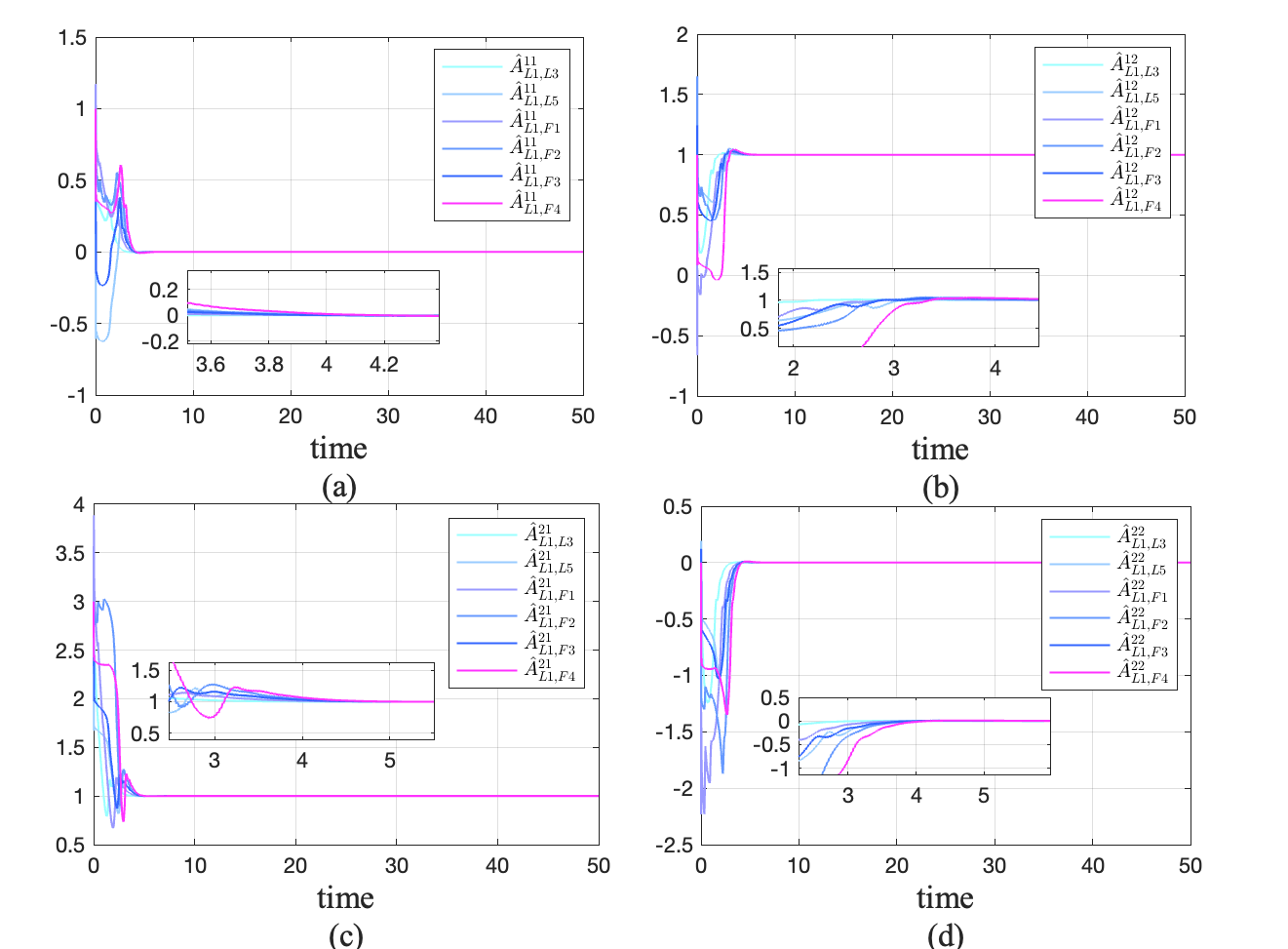}
	\caption {{Observations of agents on the formation matrix of leader L1. (a). Observation values of $1$-th row and $1$-th column. (b). Observation values of $1$-th row and $2$-th column. (c). Observation values of $2$-th row and $1$-th column. (d). Observation values of $2$-th row and $2$-th column.}}
	\label{fig5}
\end{figure}

\begin{figure}[htbp]
	\centering
	\includegraphics[scale=0.32]{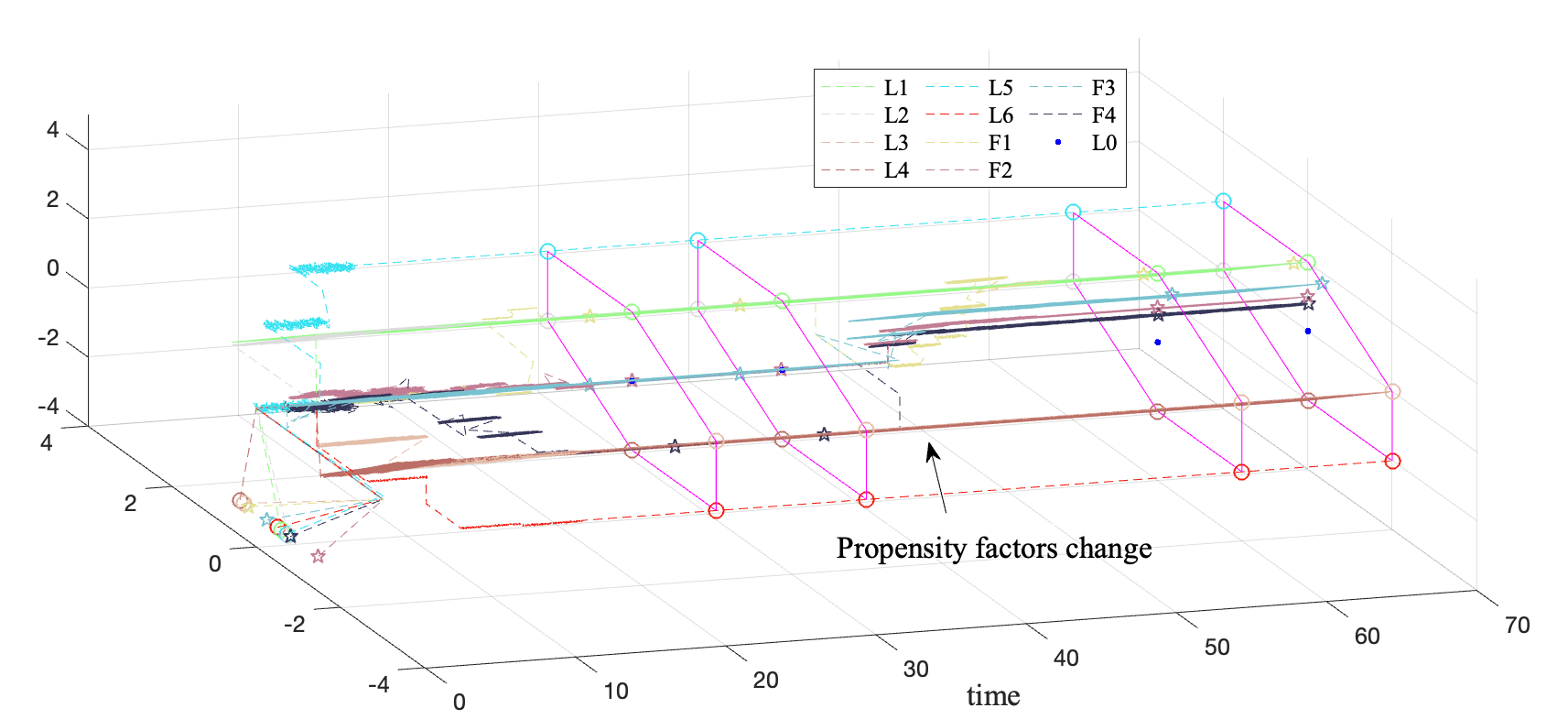}
	\caption {{3-D state trajectories of the followers and leaders for the proposed PFCC method.}}
	\label{fig7}
\end{figure}

\begin{figure}[htbp]
	\centering
	\includegraphics[scale=0.32]{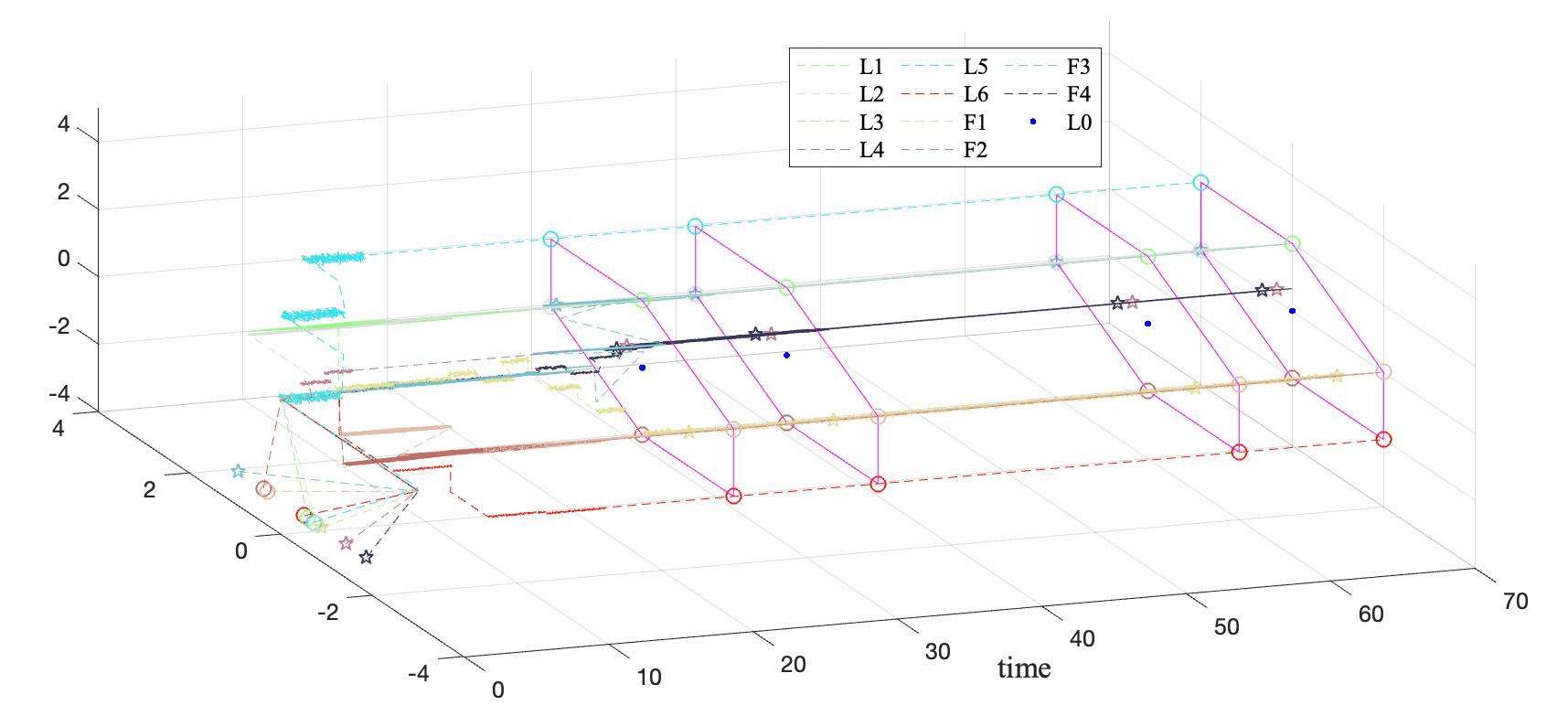}
	\caption {{3-D state trajectories of the followers and leaders for the traditional FCC method \cite{9870036}.}}
	\label{fig8}
\end{figure}
\section{Conclusions}\label{S5}
This paper addresses the Propensity Formation-Containment Control (PFCC) problem within the context of fully heterogeneous Multi-Agent Systems (MAS) that exhibit varied formation dynamics. To tackle this problem, we introduced a distributed algorithm for collecting essential control information. Additionally, we designed a data-driven adaptive observer to estimate the state of the tracking leader or the state of the leader's formation for all agents, including the Influential Transit Formation Leader (ITFL). Leveraging this acquired information and observations, we introduced a model-based control protocol that reveals the relationship between the regulation equations and control gains. This protocol guarantees the asymptotic convergence of the agents' expected values. Moreover, we developed a corresponding online data-driven learning algorithm for this control protocol, eliminating the need for any model information throughout the entire control process. The effectiveness of the proposed method was validated through a numerical example. In the future, it may be of interest to explore the relationship between the convergence rate of this method and its adjustable parameters.

\bibliography{refPFCC.bib} 
\bibliographystyle{IEEEtran} 
\end{document}